\documentclass[11pt]{book}

\usepackage[a4paper, bindingoffset=1cm, footskip=2cm, bottom=4cm, textwidth=158mm]{geometry}

\usepackage[textsize=scriptsize]{todonotes}

\usepackage[title]{appendix}

\usepackage{textgreek}  % Greek letters in bibliography
\usepackage[backend=biber, backref, style=ieee-alphabetic]{biblatex}
\bibliography{references}

\usepackage{bookmark}

\usepackage{makeidx}
\makeindex

\usepackage{amsmath, amssymb}
\usepackage{mathtools}
\usepackage{stmaryrd}  % For \llbracket and \rrbracket
\usepackage{bm}  % Bold math symbols
\usepackage{trimclip, adjustbox}  % For the \VDash newcommand

\usepackage{amsthm}
\newtheorem{theorem}{Theorem}[section]
\newtheorem{proposition}[theorem]{Proposition}
\newtheorem{lemma}[theorem]{Lemma}
\newtheorem{corollary}[theorem]{Corollary}
\newtheorem{conjecture}[theorem]{Conjecture}

\theoremstyle{definition}
\newtheorem{definition}[theorem]{Definition}

\theoremstyle{remark}
\newtheorem*{remark}{Remark}

\usepackage{ebproof}

\usepackage{alltt}  % For verbatim environments with embeddable math
\usepackage{varwidth}  % To center verbatim and alltt environments
\usepackage{textcomp}  % For \textlangle, \textrangle
\usepackage[zerostyle=d]{newtxtt}  % Dotted zeros

\usepackage{enumitem}

\newcommand{\colonequiv}{\vcentcolon\equiv}

\newcommand{\qqquad}{\qquad\qquad}
\newcommand{\turnstile}{\vdash}

\newcommand{\context}{\turnstile}

\newcommand{\red}{\triangleright}
\newcommand{\ctx}{\mathrm{ctx}}
\DeclareMathOperator{\id}{id}
\DeclareMathOperator{\ind}{ind}
\DeclareMathOperator{\refl}{refl}
\DeclareMathOperator{\isequiv}{isequiv}
\DeclareMathOperator{\idtoeqv}{idtoeqv}
\DeclareMathOperator{\transport}{transport}
\DeclareMathOperator{\univalence}{univalence}
\newcommand{\fst}{\mathrm{fst}}

\newcommand{\code}[1]{\textnormal{\texttt{\frenchspacing#1}}}
\newcommand{\sups}[1]{\textsuperscript{#1}}
\newcommand{\sub}[1]{\textsubscript{#1}}
\newcommand{\subsup}[2]{\rlap{\sub{#1}}\sups{#2}}
\newcommand{\ml}[1]{\underline{\textit{#1}}}

\newcommand{\I}{\textsc{i}}
\newcommand{\E}{\textsc{e}}
\newcommand{\R}{\textsc{r}}
\newcommand{\s}{\textsc{s}}
\newcommand{\T}{\textsc{t}}
\newcommand{\conv}{\textsc{conv}}
\newcommand{\emp}{\textsc{emp}}
\newcommand{\ext}{\textsc{ext}}
\newcommand{\vbl}{\textsc{vbl}}
\newcommand{\form}{\textsc{form}}
\newcommand{\intro}{\textsc{intro}}
\newcommand{\elim}{\textsc{elim}}

\newcommand{\cng}{\textsc{cong}}
\newcommand{\ua}{\textsc{ua}}
\newcommand{\res}{\textsc{res}}

\newenvironment{pure}{%
  \begin{center}
  \varwidth{\linewidth}%
  \begin{alltt} \frenchspacing
}{%
  \end{alltt}
  \endvarwidth
  \end{center}
}
\newcommand{\Implies}{\(\Longrightarrow\)}
\newcommand{\Equiv}{\(\equiv\)}
\newcommand{\All}{\(\bigwedge\)}
\newcommand{\Dblcolon}{\(\dblcolon\)}
\newcommand{\Colon}{\(\colon\)}
\newcommand{\Map}{\(\Rightarrow\)}
\newcommand{\llb}{\(\llbracket\)}
\newcommand{\rrb}{\(\rrbracket\)}
\newcommand{\Dotsb}{\(\dotsb\)}
\newcommand{\Dotsc}{\(\dotsc\)}
\newcommand{\Prod}[2]{\(\prod\){#1}. #2}
\newcommand{\Sum}[2]{\(\sum\){#1}. #2}
\newcommand{\V}[1]{\(\vec{#1}\)}
\newcommand{\Lt}{\(<\)}
\newcommand{\Leq}{\(\leq\)}

\newcommand{\blm}{\(\bm{\lambda}\)}
\newcommand{\rarrow}{\(\rightarrow\)}
\newcommand{\NN}{\(\mathbb{N}\)}
\newcommand{\zero}{\(\bm{0}\)}
\newcommand{\one}{\(\bm{1}\)}

\author{}
\date{}

\title{
  \textbf{An Implementation of\\Homotopy Type Theory\\in Isabelle/Pure} \\[0.5em]
  { \Large \textsc{Joshua Chen} } \\
  { \normalsize 12th September 2018 } \\[5em]
  { \normalsize A thesis for the degree of Masters in Mathematics } \\
  { \normalsize Advisor: Prof. Dr. Peter Koepke } \\
  { \normalsize Second Examiner: P.D. Dr. Philipp L\"{u}cke } \\
  { \normalsize \textsc{Mathematisches Institut} } \\[6em]
  { \normalsize \textsc{Mathematisch-Naturwissenschaftliche Fakult\"{a}t der \\
  Rheinischen Friedrich-Wilhelms-Universit\"{a}t Bonn}}
}

\begin{document}

\begin{titlepage}
\maketitle
\end{titlepage}

\newpage
\setcounter{tocdepth}{1}
\pdfbookmark[chapter]{\contentsname}{toc}
\tableofcontents

\chapter*{Introduction}
\addcontentsline{toc}{chapter}{Introduction}

This thesis presents an implementation of an initial segment of a homotopy type theory in the proof assistant Isabelle.

To elaborate, a \emph{homotopy type theory} is a type theory typically based on some version of Martin-L\"{o}f's intuitionistic type theory \cite{ml73, ml82, bibliopolis} with universes, distinguished crucially by its use of the intensional equality type and the univalence axiom, as well as by the presence of higher inductive types.
Isabelle \cite{isabelle2018} is an interactive computer proof assistant, which allows for the formalization of various object logics in a natural deduction style within a ``meta-logical framework'' called Isabelle/Pure, or simply \emph{Pure}.
The logic of Pure is itself a fragment of simple type theory \cite{church-simple-type-theory}, with the primitive logical connectives \Implies\ for entailment/implication and \Equiv\ for equality, and universal quantification \All\ over meta-level typed variables.
Finally, the goal and major contribution of this thesis is to formalize within the Pure logic a version of homotopy type theory that corresponds suitably to that developed in the Homotopy Type Theory book \cite{hottbook}, and also to implement semi-automated routines for constructing and verifying proofs within this theory.
We implement an ``initial segment'' of the theory consisting of the universe hierarchy, the basic types, and the univalence axiom.
$W$-types and higher inductive types are left for future work.

In formalizing one formal system $\mathcal{L}$ (the \emph{object logic}) within another $\mathcal{F}$ (the \emph{metalogic/logical framework}), there are often significant differences in the structure and expressive ability\footnote{By which is meant the \emph{convenience} with which one can express certain logical constructions, as opposed to the \emph{proof-theoretic strength} of the two logics; it may be the case that both $\mathcal{L}$ and $\mathcal{F}$ can prove the same statements, but one in a much simpler and clearer way than the other.
An analogy is the contrast between Conway's Game of Life and the C++ programming language; both are Turing complete, but there is a reason software is not typically written using cellular automata.}
of the object- and meta- logics, which creates issues when one comes to implement one inside the other.
For example, the standard formalism of the type theory presented in \cite{hottbook} requires the explicit derivation of contexts to be used in hypothetical typing judgments, but there is no simple way of handling object-level contexts in Pure.
This results in the impossibility of directly translating the type rules in a way which preserves the formal logical properties discussed above, and we are forced to make design decisions as to good alternative formulations, which must still deviate in some way from the original intended logic.
Much of the work of this thesis consists in explaining and justifying these design decisions, and showing that the resulting modified implementation still corresponds well to the original theory.

The guiding goal of the implementation described in this thesis has been to serve as a formalization of some of the content of \cite{hottbook}, and an aid in exploring the theory described therein.
It is however also hoped that this initial attempt at bringing homotopy type theory to the Isabelle prover will lead to the development of a more powerful and robust library for the logic, perhaps one that will ultimately be comparable to those of systems like Coq and Agda.

\chapter*{Notation and conventions}
\addcontentsline{toc}{chapter}{Notation and conventions}

The following notation and conventions will be introduced throughout the course of this thesis.
We gather them here for ease of reference.

\section{Orthography}

\begin{itemize}
\item Theory names are written in italics, e.g. $ITT$, $HoTT$.

\item Names of axioms and inference rules are given in caps, e.g. $\ua$ or UA for the univalence axiom.

\item Names of entities of the Isabelle/Pure framework are written in monospace font, e.g. \code{prop} for the type of Pure formulas, or \code{Prod\_intro} for the Pure implementation of the dependent function introduction rule.

\item Names of ML entities are denoted with underlined italic text, e.g. \ml{thm} for the ML type of Pure theorems, \ml{Var} for the ML term constructor of Pure schematic variables.
\end{itemize}

\section{Notation}

\begin{itemize}[label=]
\item $\mathcal{P}_f(A)$ denotes the finite powerset of $A$.
\item $\mathcal{D}_\mathcal{N}$ denotes the deductive core of a natural deduction system $\mathcal{N}$.
\item \code{t}$[\sigma]$ is the instantiation of the term \code{t} by the higher-order unifier $\sigma$.
\end{itemize}

\section{Operator precedences}

In descending order; entries in the same item have equal precedence and must be disambiguated with parentheses:

\begin{enumerate}
\item $\centerdot$ (path composition) % 120
\item $\circ$ (function composition) % 110
\item $=$ (equality), $\sim$ (homotopy), $\simeq$ (equivalence) % 100
\item $\times$ (nondependent sum) % 50
\item $\rightarrow$ (nondependent product) %40
\item $\equiv$ (definitional equality) % 2
\end{enumerate}

Juxtaposition of terms has higher precedence than definitional equality, so
\begin{pure}
f x \Equiv g x
\end{pure}
means
\begin{pure}
(f x) \Equiv (g x)\textrm{.}
\end{pure}
Binders, e.g. $\prod$, $\sum$, $\lambda$, are greedy and continue to the end of their innermost enclosing scope.
Arrows \Implies, $\rightarrow$ and function composition $\circ$ associate to the right.
The dependent sum constructor $\times$ and path composition $\centerdot$ associate to the left.

\chapter{Homotopy type theory} \label{ch:1}

% Start page numbering of primary content
\setcounter{page}{1}
\pagenumbering{arabic}

We begin with a brief discussion of the homotopy type theory $HoTT$ we aim to implement, and then provide a formal description in Sections \ref{sec:ITT} and \ref{sec:ua}.
$HoTT$ is a fragment of the type theory considered in the Homotopy Type Theory (HoTT) book \cite{hottbook}, consisting of the dependent product, dependent sum, empty, unit, coproduct, natural number and equality types, together with a Russell-style universe hierarchy and the univalence axiom.

This chapter elaborates on the presentation given in Appendix A.2 of \cite{hottbook}.
Although based on Per Martin-L\"{o}f's intuitionistic type theories \cite{ml73, ml82, bibliopolis}, an additional judgment is introduced for the well-formedness of contexts, which were handled more implicitly in the original development.

It should be noted that there is not one single ``homotopy type theory''---we have chosen a particular system, but there are other ways to formalize and extend the homotopy interpretation of intensional type theory.\footnote{Cubical type theory \cite{cubicaltt} is a particularly promising approach in which univalence is a theorem, as opposed to our theory, in which it has to be taken as an axiom.}

\section{Terms, types, and contexts} \label{sec:terms-types-contexts}

As usual for a type theory, we have two syntactic classes of expressions: \emph{terms} and \emph{types}.

Types are introduced via \emph{formation rules}.
The terms are generated from countably many variable and constant symbols, augmented with additional constructors associated with particular types via \emph{introduction rules}.

Terms $a$ may be assigned types $A$, and as usual we write $a \colon A$ to express that $a$ is assigned (``has'') type $A$.

\begin{definition}
A \textbf{context} is a list of variable symbol typing assignments
\[ x_1 \colon A_1, \dotsc, x_n \colon A_n \]
for some $n \geq 0$ (i.e. it is allowed to be empty).
\end{definition}

Contexts are used as ``type assumptions'' in the judgments of the theory.
Each type $A_i$ appearing in the context is allowed to depend on the variables $x_1, \dotsc, x_{i-1}$ that appear before it, but not on any variables appearing after.

\section{Judgments} \label{sec:judgments}

Judgments are the ``propositions'' of our logical theory---they are the statements which the logic is allowed to prove.

There are three judgment forms
\[ \Gamma\ \ctx, \qquad
\Gamma \context a \colon A, \qquad
\Gamma \context a \equiv b \colon A, \]
where $\Gamma$ is a context, $a$ and $b$ are terms, and $A$ is a type.
They express, respectively, that $\Gamma$ is a well-formed context, that term $a$ has type $A$ under the type assumptions in $\Gamma$, and that $a$ and $b$ are definitionally equal terms of type $A$ under the type assumptions in $\Gamma$.

Definitional equality can be seen as a meta-level notion of ``symmetric reduction''.
It subsumes $\beta$- and $\eta$-conversion but not $\alpha$-conversion, which is handled on the meta-level.\footnote{To be more explicit, standard formulations of type theory typically take the changing of bound variable names in a term expression as a meta-level operation preserving the syntactic identity of the term (subject to side conditions forbidding variable capture).
That is, technically speaking, the class of terms is considered modulo the equivalence relation generated by changing bound variable symbols uniformly throughout a term, subject to no-variable-capture.
This meta-level syntactic identity may be applied at any time, and is not handled by any of the rules governing definitional equality.}

\begin{remark}
Martin-L\"{o}f also considered the judgment forms
\[ A\ \mathrm{type}, \qquad
A \equiv B\ \mathrm{type}, \]
judging $A$ to be a type and $A$, $B$ to be definitionally equal types.
We do not need these in our theory due to our use of universe types.
\end{remark}

\section{Universes} \label{sec:universes}

Intuitively, a universe type is a type whose inhabitants are themselves types.
There are two ways to interpret this statement.

The first is to consider the inhabitants $A$ of a universe type $U$ to be \emph{names} for actual types, and to introduce an interpretation operator $\mathrm{El}$ taking every name $A \colon U$ to a type $\mathrm{El}(A)$.
This is the \emph{Tarski-style} formulation.
It does not conflate the syntactic categories of terms and types and has nice meta-theoretic properties, but is more cumbersome to work with, especially when we come to introduce a hierarchy of universes.

The second approach and the one we take is the \emph{Russell-style} formulation.
Here a term $A \colon U$ is also an actual type that may be used in typing assignments, so that there may be terms $a \colon A$.
Note that in this formulation, the types are a proper subset of the terms.

In $HoTT$ we consider not just one universe type, but introduce a countable hierarchy
\[ U_0, U_1, \dotsc \]
such that each term $U_i$ is an inhabitant of the next universe, i.e. we assign
\[ U_i \colon U_{i+1} \]
for all natural numbers $i$.
Note that the indices of the universes are \emph{numerals} i.e. meta-level natural numbers.
In addition, we require the hierarchy to be \emph{cumulative}: if $A \colon U_i$ for some $i$, then also $A \colon U_j$ for all $j > i$.

\begin{remark}
It is known that while Russell-style universes seem more straightforward to use in practice, they can suffer from meta-theoretic issues, particularly in relation to canonicity and subject reduction \cite{luo12}.
We discuss these issues further in Section \ref{sec:metatheory}.
\end{remark}

\section{Rules for $ITT$} \label{sec:ITT}

We first formally state the axioms and inference rules for $ITT$ (intuitionistic type theory), the fragment of $HoTT$ without the univalence rule.
These hold for all numerals $i$.

\subsection{Structural rules}

\noindent \emph{Contexts:}

\[
\begin{prooftree}[center=false] \label{rule:ctx-emp-ext}
  \infer0[ctx-\emp]{ \cdot\ \ctx }
\end{prooftree}
\qqquad
\begin{prooftree}[center=false]
  \hypo{ \Gamma \context A \colon U_i }
  \infer1[ctx-\ext]{ \Gamma, x \colon A\ \ctx }
\end{prooftree}
\]
where $x$ is a variable symbol that does not appear in $\Gamma$.

\[
\begin{prooftree} \label{rule:vbl}
  \hypo{ x_1 \colon A_1, \dotsc, x_n \colon A_n\ \ctx }
  \infer1[\vbl]{ x_1 \colon A_1, \dotsc, x_n \colon A_n \context x_j \colon A_j }
\end{prooftree}
\]
for any $1 \leq j \leq n$.
As discussed in Section \ref{sec:terms-types-contexts}, a type $A_i$ appearing in a well-formed context may depend on the variables $x_j$ for $1 \leq j < i$, and never on the variables $x_j$ for $j \geq i$.
\bigskip

\noindent \emph{Substitution and weakening:}

\[
\begin{prooftree}
  \hypo{\Gamma \context a \colon A}
  \hypo{\Gamma, x \colon A, \Delta \context b \colon B}
  \infer2[\textsc{subst}$_1$]{\Gamma, \Delta[a/x] \context b[a/x] \colon B[a/x]}
\end{prooftree}
\]

\[
\begin{prooftree}
  \hypo{\Gamma \context a \colon A}
  \hypo{\Gamma, x \colon A, \Delta \context b \equiv c \colon B}
  \infer2[\textsc{subst}$_2$]{\Gamma, \Delta[a/x] \context b[a/x] \equiv c[a/x] \colon B[a/x]}
\end{prooftree}
\]

\[
\begin{prooftree}[center=false]
  \hypo{\Gamma \context A \colon U_i}
  \hypo{\Gamma, \Delta \context b \colon B}
  \infer2[\textsc{wkg}$_1$]{\Gamma, x \colon A, \Delta \context b \colon B}
\end{prooftree}
\qqquad
\begin{prooftree}[center=false]
  \hypo{\Gamma \context A \colon U_i}
  \hypo{\Gamma, \Delta \context b \equiv c \colon B}
  \infer2[\textsc{wkg}$_2$]{\Gamma, x \colon A, \Delta \context b \equiv c \colon B}
\end{prooftree}
\]
\bigskip

\noindent \emph{Equivalence and structural congruence rules:}

\[
\begin{prooftree}[center=false]
  \hypo{\Gamma \context a \colon A}
  \infer1[$\equiv$-\textsc{refl}]{\Gamma \context a \equiv a \colon A}
\end{prooftree}
\qqquad
\begin{prooftree}[center=false] \label{rule:equiv-sym}
  \hypo{\Gamma \context a \equiv b \colon A}
  \infer1[$\equiv$-\textsc{sym}]{\Gamma \context b \equiv a \colon A}
\end{prooftree}
\]

\[
\begin{prooftree}
  \hypo{\Gamma \context a \equiv b \colon A}
  \hypo{\Gamma \context b \equiv c \colon A}
  \infer2[$\equiv$-\textsc{trans}]{\Gamma \context a \equiv c \colon A}
\end{prooftree}
\]

\[
\begin{prooftree}[center=false]
  \hypo{\Gamma \context a \colon A}
  \hypo{\Gamma \context A \equiv B \colon U_i}
  \infer2[\textsc{type-cong}]{\Gamma \context a \colon B}
\end{prooftree}
\qqquad
\begin{prooftree}[center=false]
  \hypo{\Gamma \context a \equiv b \colon A}
  \hypo{\Gamma \context A \equiv B \colon U_i}
  \infer2[\textsc{def-cong}]{\Gamma \context a \equiv b \colon B}
\end{prooftree}
\]

\subsection{Universes}

\[
\begin{prooftree}[center=false]
  \hypo{\Gamma\ \ctx}
  \infer1[$U_i$-\intro]{\Gamma \context U_i \colon U_{i+1}}
\end{prooftree}
\qqquad
\begin{prooftree}[center=false]
  \hypo{\Gamma \context A \colon U_i}
  \infer1[$U_i$-\textsc{cumul}]{\Gamma \context A \colon U_{i+1}}
\end{prooftree}
\]

\subsection{Type rules}

These are the formation, introduction, elimination, computation and uniqueness rules as given in Sections A.2.4--A.2.10 of \cite{hottbook}.

\subsection{Term congruence rules}

The type rules introduce constants which are built from term constructors taking some number of arguments.
A congruence rule expresses the principle that a term constructor applied to definitionally equal arguments yields definitionally equal terms.
We state the rules for dependent sum and product here; the rules for the other types are analogous.

Dependent product:

\[
\begin{prooftree}
  \hypo{\Gamma \context A \equiv A^\prime \colon U_i}
  \hypo{\Gamma, x \colon A \context B \equiv B^\prime \colon U_i}
  \infer2[$\Pi$-\form-\cng]{\Gamma \context \prod_{x \colon A} B \equiv \prod_{x \colon A^\prime} B^\prime \colon U_i}
\end{prooftree}
\]

\[
\begin{prooftree}
  \hypo{\Gamma \context A \equiv A^\prime \colon U_i}
  \hypo{\Gamma, x \colon A \context b \equiv b^\prime \colon B}
  \infer2[$\Pi$-\intro-\cng]{\lambda(x \colon A). b \equiv \lambda(x \colon A^\prime). b^\prime \colon \prod_{x \colon A} B}
\end{prooftree}
\]

\[
\begin{prooftree}
  \hypo{\Gamma \context f \equiv f^\prime \colon \prod_{x \colon A} B}
  \hypo{\Gamma \context a \equiv a^\prime \colon A}
  \infer2[$\Pi$-\elim-\cng]{\Gamma \context f(a) \equiv f^\prime(a^\prime) \colon B[a/x]}
\end{prooftree}
\]
\smallskip

Dependent sum:

\[
\begin{prooftree}
  \hypo{\Gamma \context A \equiv A^\prime \colon U_i}
  \hypo{\Gamma, x \colon A \context B \equiv B^\prime \colon U_i}
  \infer2[$\Sigma$-\form-\cng]{\Gamma \context \sum_{x \colon A} B \equiv \sum_{x \colon A^\prime} B^\prime \colon U_i}
\end{prooftree}
\]

\[
\begin{prooftree}
  \hypo{\Gamma, x \colon A \context B \colon U_i}
  \hypo{\Gamma \context a \equiv a^\prime \colon U_i}
  \hypo{\Gamma \context b \equiv b^\prime \colon B[a/x]}
  \infer3[$\Sigma$-\intro-\cng]{\Gamma \context (a,b) \equiv (a^\prime, b^\prime) \colon \sum_{x \colon A} B}
\end{prooftree}
\]

\[
\begin{prooftree}
  \hypo{
  \begin{gathered} \textstyle
    \Gamma, z \colon \sum_{x \colon A} B \context C \equiv C^\prime \colon U_i
    \qquad
    \Gamma, x \colon A, y \colon B \context g \equiv g^\prime \colon C[(x,y)/z] \\ \textstyle
    \Gamma \context p \equiv p^\prime \colon \sum_{x \colon A} B
  \end{gathered}
  }
  \infer1[$\Sigma$-\elim-\cng]{\Gamma \context \ind_{\sum_{x \colon A} B}(z.C, x.y.g, p) \equiv \ind_{\sum_{x \colon A} B}(z.C^\prime, x.y.g^\prime, p^\prime) \colon C[p/z]}
\end{prooftree}
\]
\smallskip

\section{Univalence and $HoTT$} \label{sec:HoTT}

Thus far the theory has simply been that of intuitionistic Martin-L\"{o}f type theory.
Here we introduce the final ingredient of $HoTT$---the univalence axiom.

Univalence is the first major homotopy-type-theoretic idea, so we pause to reiterate the foundational concepts and results.
We assume familiarity with the basic terminology and definitions of \cite{hottbook}, particularly concerning type families, the equality type, and path induction.

\begin{definition}
Let $f, g \colon \prod_{x \colon A} B$.
Let
\[ f \sim g \colonequiv \prod_{x \colon A} f(x) =_{B(x)} g(x) \]
be the type of \textbf{homotopies between $f$ and $g$}.
\end{definition}

\begin{definition} \label{def:isequiv}
Let $f \colon A \rightarrow B$.
Define
\[
\isequiv(f) \colonequiv \Bigg( \sum_{g \colon B \rightarrow A} g \circ f \sim \id_A \Bigg) \times \Bigg( \sum_{g \colon B \rightarrow A} f \circ g \sim \id_B \Bigg)
\]
where $g \circ f \colonequiv \lambda(x \colon A).\,g(f(x))$ and $\id_A \colonequiv \lambda(x \colon A).\,x$.
If $\isequiv(f)$ is inhabited we call $f$ an \textbf{equivalence}.
\end{definition}

\begin{definition}
Let $A, B \colon U_i$.
Let
\[ A \simeq B \colonequiv \sum_{f \colon A \rightarrow B} \isequiv(f) \]
be the type of \textbf{equivalences} from $A$ to $B$.
\end{definition}

That is, a homotopy between dependent functions $f \sim g$ expresses \emph{pointwise equality} of $f$ and $g$; a function $f$ is an equivalence if it has a pointwise left inverse and a pointwise right inverse, not necessarily identical; and two types $A$ and $B$ are equivalent if there is an equivalence $f \colon A \rightarrow B$.

Equivalence of types is meant to express the notion of \emph{isomorphism} between mathematical structures, encoded as types.
With this in mind, the following result is fitting:

\begin{theorem}[Equal types are equivalent] \label{thm:idtoeqv}
Let $A, B \colon U_i$.
There is a function
\[ \idtoeqv_{A,B} \colon A =_{U_i} B \rightarrow A \simeq B. \]
\end{theorem}

The proof of this theorem is straightforward and given in \cite{hottbook}; we reproduce it more explicitly here in order to facilitate the formal statement of the univalence axiom later.
First we need the following lemma.

\begin{lemma}[Transport lemma] \label{lem:transport}
Given a type $A \colon U_i$, a type family $P \colon A \rightarrow U_i$, $x, y \colon A$ and $p \colon x =_A y$, we can define a function $\transport^P_{x,y,p} \colon P(x) \rightarrow P(y)$ satisfying
\[
\transport^P_{x,x,\refl(x)} \equiv \id_{P(x)}
\]
for all $x \colon A$.
\end{lemma}

\begin{proof}
By path induction on $x,y,p$.
Let $C(\widetilde{x},\widetilde{y},\widetilde{p}) \colonequiv P(\widetilde{x}) \rightarrow P(\widetilde{y})$.
It suffices to exhibit $c(\widetilde{x}) \colon C(\widetilde{x},\widetilde{x},\refl(\widetilde{x})) \equiv P(\widetilde{x}) \rightarrow P(\widetilde{x})$ for every $\widetilde{x} \colon A$, from which it will follow that
\[ \ind_=(\widetilde{x}.\widetilde{y}.\widetilde{p}.C, \widetilde{x}.c, x, y, p) \colon C(x,y,p) \equiv P(x) \rightarrow P(y) \]
for all $x, y \colon A$ and $p \colon x =_A y$.
Defining $c(\widetilde{x}) \colonequiv \id_{P(\widetilde{x})}$ for all $\widetilde{x}$ gives such a function.

Define
\[ \transport^P_{x,y,p} \colonequiv \ind_=(\widetilde{x}.\widetilde{y}.\widetilde{p}.C, \widetilde{x}.c, x, y, p), \]
then
\begin{align*}
\transport^P_{x,x,\refl(x)} & \equiv \ind_=(\widetilde{x}.\widetilde{y}.\widetilde{p}.C, \widetilde{x}.c, x, x, \refl(x)) \\
& \equiv c(x) \\
& \equiv \id_{P(x).}
\end{align*}
\end{proof}

\begin{proof}[Proof of Theorem \ref{thm:idtoeqv}]
Let $p \colon A =_{U_i} B$ be given, and consider
\[ \transport^{\id_{U_i}}_{A,B,p} \colon A \rightarrow B. \]
We claim that this is an equivalence, so that the desired function is
\begin{equation}
\idtoeqv_{A,B} \colonequiv \lambda(p \colon A =_{U_i} B).\,(\transport^{\id_{U_i}}_{A,B,p}, \mathfrak{p}), \label{eq:idtoeqv}
\end{equation}
where $\mathfrak{p} \colon \isequiv(\transport^{\id_{U_i}}_{A,B,p})$.

To construct such a term $\mathfrak{p}$ we use path induction on $A,B,p$.
Let
\[ C(\widetilde{A}, \widetilde{B}, \widetilde{p}) \colonequiv \isequiv(\transport^{\id_{U_i}}_{\widetilde{A},\widetilde{B},\widetilde{p}}). \]
We have
\[ \isequiv(\transport^{\id_{U_i}}_{\widetilde{A},\widetilde{B},\widetilde{p}}) \equiv \isequiv(\id_{\widetilde{A}}) \]
and
\[ c(\widetilde{A}) \colonequiv \big( (\id_{\widetilde{A}}, \lambda(x \colon \widetilde{A}).\,\refl(x)), (\id_{\widetilde{A}}, \lambda(x \colon \widetilde{A}).\,\refl(x)) \big) \colon \isequiv(\id_{\widetilde{A}}) \]
for all $\widetilde{A} \colon U_i$.
Thus
\[ \mathfrak{p} \colonequiv \ind_=(\widetilde{A}.\widetilde{B}.\widetilde{p}.C, \widetilde{A}.c, A, B, p) \colon \isequiv(\transport^{\id_{U_i}}_{A,B,p}). \]
\end{proof}

\subsection{The univalence axiom} \label{sec:ua}

In the previous paragraphs we proved that equal types are equivalent, or isomorphic.
The converse, that equivalent types are equal, is however not provable.
It is instead given by the \emph{univalence axiom} which in fact states slightly more: postulating that the function $\idtoeqv$ is itself an equivalence, and hence ``invertible''.
As a formal inference rule, it states:

\[
\begin{prooftree} \label{rule:UA}
  \hypo{\Gamma \context A \colon U_i}
  \hypo{\Gamma \context B \colon U_i}
  \infer2[\ua]{\Gamma \context \univalence(A,B) \colon \isequiv(\idtoeqv_{A,B})}
\end{prooftree}
\]
where $\isequiv$ and $\idtoeqv$ are as given in Definition \ref{def:isequiv} and \eqref{eq:idtoeqv}.

Define
\[ HoTT = ITT + \text{UA}, \]
that is, the homotopy type theory we consider is $ITT$ plus the univalence axiom schema.

\section{Metatheory} \label{sec:metatheory}

When considering a type theory, we often ask if the following problems are decidable.
In the following, $\Gamma$ is some context.

\begin{enumerate}
  \item \textbf{Definitional equality}: given two terms $a$ and $b$, determine if $\Gamma \context a \equiv b \colon A$ for some type $A$.
  \item \textbf{Type checking}: given a term $a$ and a type $A$, determine if $\Gamma \context a \colon A$.
  \item \textbf{Type inference}: given a term $a$, find a type $A$ such that $\Gamma \context a \colon A$.
\end{enumerate}

In \cite{ml73} a type theory which we will call $MLTT_{73}$ is presented and shown to have a normalization theorem, which then implies decidability of all three problems above \emph{for well-typed terms}.
The proof there is given for the empty context, but it is noted that it can be modified to accommodate open terms as well.

\subsection{Metatheory of $ITT$}

In this section we consider only $ITT$.
$ITT$ is very similar to $MLTT_{73}$, albeit different in a few ways.
Its universes are cumulative, so that not every term has a unique type up to definitional equivalence.
The rules for $\Pi$ and $\Sigma$ type formation (called reflection in \cite{ml73}) are slightly different to account for the different universe formulations.
Also, the definitional equality judgment in $ITT$ is typed while that of $MLTT_{73}$ is not.
Despite these differences, we can carry the proof of normalization over to $ITT$ with a small modification of the definitions.

\begin{definition}
Define the \textbf{reduction relation}
\[ \Gamma \context a \red b \]
on terms $a, b$ to be the partial order generated by the rules governing definitional equality $\Gamma \context a \equiv b \colon A$, with the exception of the symmetry rule \hyperref[rule:equiv-sym]{$\equiv$-\textsc{sym}}.
\end{definition}

That is to say, the rules governing $\Gamma \context a \red b$ are the reflexive, transitive, substitution, structural congruence, type computation, and term congruence rules, but replacing $\equiv$ with $\red$.

Note that the reduction relation does not specify \emph{a priori} that $a$ and $b$ have the same type.
In particular, it needs to be proved that $\Gamma \context a \red b$ and $\Gamma \context a \colon A$ implies $\Gamma \context b \colon A$ (and the analogous statement for definitional equality).
In fact this will be a consequence of the normalization theorem.

\begin{theorem}[Normalization theorem for $ITT$]
There is a normal form $a^\ast \colon A^\ast$ defined on well-typed terms $a \colon A$, together with an algorithm that gives, for every well-typed closed term $a \colon A$, a normal term $a^\ast \colon A^\ast$ such that $\context a \red a^\ast$.
\end{theorem}

\begin{proof}[Proof sketch]
The normal form on closed well-typed terms is defined inductively:
\begin{itemize}
\item Any nullary constructor $c \colon C$ given by a type formation or introduction rule is a normal term with corresponding \emph{normal type} $C$.
Formation and introduction rules are those rules ending in -\form\ or -\intro, respectively.

\item If $c(x_1, \dotsc, x_n) \colon A(x_1, \dotsc, x_n)$ is a constructor with free variables
\[ x_1 \colon A_1, x_2 \colon A_2(x_1), \dotsc, x_n \colon A_n(x_1, \dotsc, x_{n-1}) \]
given by the conclusion of a formation or introduction rule,
\[ c_1 \colon C_1, c_2 \colon C_2, \dotsc c_n \colon C_n \]
are normal terms with corresponding normal types, and for each $1 \leq i \leq n$ either
  \begin{enumerate}[label=(\roman*)]
    \item $\context A_i[c_1, \dotsc, c_{i-1}/x_1, \dotsc, x_{i-1}] \red C_i$, or
    \item there are numerals $k \leq l$ such $C_i$ is syntactically identical to $U_k$ and
    \[ \context A_i[c_1, \dotsc, c_{i-1}/x_1, \dotsc, x_{i-1}] \red U_l, \]
  \end{enumerate}
then
\[ c[c_1, \dotsc, c_n/x_1, \dotsc, x_n] \colon A_n[c_1, \dotsc, c_n/x_1, \dotsc, x_n] \]
is a normal term with corresponding normal type.
\end{itemize}

Observe by induction that the normal terms as we have defined them are indeed \emph{normal}, that is, they do not reduce nontrivially.

Every normal term has a unique corresponding normal type.
The proof will show that the normal types $A^\ast$ corresponding to normal terms $a^\ast$ as given by the inductive definition agree with the normal form of the type $A \colon U_i$ considered as a term, where $U_i$ is the smallest universe containing $A$.
That is to say, if $a \colon A$ has normal form $a^\ast \colon A^\ast$, then $A \colon U_i$ has normal form $A^\ast \colon U_i$.

With the above definition of normal terms, the proof of the normalization theorem in Section 3.3 of \cite{ml73} goes through \emph{mutatis mutandis} for $ITT$.
The core idea is to show by induction on the length of a derivation of $\context a \colon A$ that there is a normal term $a^\ast \colon A^\ast$ such that $\context a \red a^\ast$, and by induction on the length of a derivation of $\context a \equiv b \colon A$ that $a^\ast$ and $b^\ast$ are syntactically identical.
\end{proof}

\begin{corollary}[Syntactic uniqueness of normal forms]
If $a \colon A$ and $b \colon A$ are closed normal terms with $\context a \equiv b \colon A$, then $a$ and $b$ are syntactically identical.
\end{corollary}

\begin{proof}
Since $a$ and $b$ are closed normal terms they are syntactically identical to their normal forms $a^\ast$ and $b^\ast$ respectively.
The result follows from the last sentence of the proof sketch above and the derivation of $\context a \equiv b \colon A$.
\end{proof}

This means that the closed normal terms form a collection of unique representatives for every $\equiv$-equivalence class of well-typed terms.
We are now able to prove decidability of the problems considered at the beginning of this section.

\begin{theorem}
Definitional equality for well-typed terms is decidable.
\end{theorem}

\begin{proof}
We can reduce any two well-typed terms $a \colon A$ and $b \colon B$ to their normal forms $a^\ast \colon A^\ast$ and $b^\ast \colon B^\ast$.
By syntactic uniqueness, $\context a \equiv b \colon A$ if and only if $a^\ast$ is syntactically identical to $b^\ast$ and $A^\ast$ is syntactically identical to $B^\ast$, which we can clearly check.
\end{proof}

We observe that this means that if $\Gamma \context a \red b$ for well-typed terms $a$ and $b$, then $\Gamma \context a \equiv b \colon A$ for some type $A$.

\begin{theorem}
It is possible to check if a well-typed term is a type.
Type checking and type inference for well-typed terms are decidable.
\end{theorem}

\begin{proof}
Checking if a well-typed term $A$ is a type is solved by reducing $A$ to its normal form and observing its head constant.

Given a well-typed term $a$ and a well-formed type $A$, we check if $a \colon A$ by reducing $a$ to its normal form $a^\ast \colon B$, reducing $A$ to its normal form $A^\ast \colon U_i$, and then checking if $B$ and $A^\ast$ are syntactically identical.

To infer a type for a well-typed term $a$, we simply reduce to normal form to obtain $a^\ast \colon A^\ast$.
Then $a \equiv a^\ast \colon A^\ast$ and one can show that this implies $a \colon A^\ast$.
\end{proof}

Finally $ITT$ is consistent, in that there is no proof of contradiction from the empty context.

\begin{theorem}
In $ITT$, there is no derivable judgment $\context a \colon \bm{0}$.
\end{theorem}

\begin{proof}
Otherwise $a$ would reduce to a normal term $a^\ast \colon \bm{0}$.
But since $\bm{0}$ has no introduction rule, there is no associated constructor and hence no normal term of type $\bm{0}$.
\end{proof}

The following notions are closely related to the normalization and decidability of definitional equality.

\begin{definition}
A \textbf{canonical term} $a \colon A$ is one generated only from constructors given by the introduction rules for $A$.
Note that we view the type formation rules as introduction rules for the respective universe types.
A type theory satisfies \textbf{canonicity} if for every type $A$ and term $a \colon A$, $a$ is definitionally equal to a canonical term in $A$.
\end{definition}

\begin{definition}
A type theory with a reduction relation $\twoheadrightarrow$ on terms satisfies \textbf{subject reduction} if whenever $a \colon A$ and $a \twoheadrightarrow b$, we also have $b \colon A$.
\end{definition}

Using the normalization theorem, the following theorem is straightforward.

\begin{theorem} \label{thm:ITT-canonicity-subj-red}
$ITT$ satisfies canonicity.
With the reduction relation $\red$, it also satisfies subject reduction.
\end{theorem}

\begin{remark}
In Section \ref{sec:universes} we remarked that cumulative Russell-style universes may cause issues with canonicity and subject reduction.
For example, in Section 3 of \cite{luo12} examples are given showing the failure of both properties for a particular type theory.
However we note that both the examples given there use features that are either not present, or formulated differently in $ITT$.
\end{remark}

\begin{remark}
The results we have stated in this section hold for well-typed terms.
It is further known that there are algorithms for type checking arbitrary terms in type theories closely related to $MLTT_{73}$; see for example \cite{coquand-type-checking, magnusson95}.
However it is unclear to the author if these algorithms adapt to the specific case of $ITT$.
\end{remark}

\subsection{Metatheory of $HoTT$ (and variants)}

We have seen that the univalence axiom expresses a metamathematical property we would like our type theory to possess, namely that isomorphic structures are equal.
We have also seen that the foundational theory $ITT$ we have chosen has a normalization property, and as a consequence, is consistent and has decidable judgments for all well-typed terms.
Unfortunately, it is for the moment unclear if all these desirable properties can be merged into $HoTT$.

The univalence axiom schema \hyperref[rule:UA]{UA} introduces a non-canonical term constructor $\univalence(\cdot, \cdot)$ into specific $\Sigma$-types, which breaks canonicity and normalization as defined for $ITT$.
It appears to still be open if there exists a modified normalization process for $HoTT$ (and the theory in \cite{hottbook}).

Metatheoretic consequences of the univalence axiom together with Martin-L\"{o}f type theory in general have been considered in other works, for example in \cite{simplicial-model-uf} where it is shown that a variant of homotopy type theory with one Tarski-style universe is consistent relative to
\[ \mathrm{ZFC} + \text{ there exist two strong inaccessible cardinals } \alpha, \beta. \]
More recently, Coquand et al. \cite{cubicaltt} have constructed models of a \emph{cubical type theory} which uses Russell-style universes and \emph{path types} in place of the standard equality type, in which univalence is a theorem and canonicity for the natural numbers holds \cite{huber17}.
In this theory the standard intensional equality type can be defined in terms of the path type.

\chapter{The Isabelle/Pure framework} \label{ch:2}

We aim to implement a working version of $HoTT$ inside the interactive computer proof assistant Isabelle \cite{isabelle2018}.
Isabelle is a \emph{generic} proof assistant, in that it supports the usage of many different logics of which, loosely speaking, any one may be ``loaded'' for a given proof session.
Currently the largest and most developed logic is \emph{higher-order logic} or HOL, although there are also implementations of first-order logic and Zermelo-Fraenkel set theory, extensional Martin-L\"{o}f type theory, modal logic and various sequent calculi, etc.
Each of these \emph{object logics} is implemented on top of Isabelle's base ``meta-logical framework'', which is called Isabelle/Pure or simply \emph{Pure}.

Pure provides both an extensible logic in which the terms and formulas of a given object logic can be formalized, as well as a set of procedures for manipulating and proving such formalized statements.
In this chapter we give an introduction to both aspects of the framework.
This introduction is not intended to be a complete description; we have omitted parts (e.g. type classes, locales etc.) that are not directly used in our implementation of $HoTT$.\footnote{Readers interested in knowing more are referred to Section 2.1 of \cite{isar-ref} Chapters 2, 4 and 5 of \cite{implementation}, and Part I of \cite{old-intro} as starting points.}

\section{Formal deductive systems}

Before introducing Pure, it will help to make explicit the framework we have been implicitly using to formulate logical systems like $HoTT$.
The definitions in this section are based on Ch.\ 6 of \cite{lambda-calculus} and Ch.\ 4 of \cite{grabmeyer05}.

\subsection{Hilbert-style deductive systems} \label{sec:deductive-systems}

\begin{definition} \label{def:deductive-system}
A (Hilbert-style) \textbf{deductive system} is a triple
\[ \mathcal{D} = (\mathcal{L}, \mathcal{F}, \mathcal{R}) \]
consisting of the \emph{language} of the system, its \emph{formulas}, and its \emph{inference rules}.
The language $\mathcal{L}$ is a set with distinguished subsets called \emph{syntactic categories}.
The set of formulas $\mathcal{F}$ is a set of well-formed expressions typically defined inductively over the elements of $\mathcal{L}$.
An inference rule $R \in \mathcal{R}$ is given by a set $R = \left\{(\Phi_\alpha, \phi_\alpha) \,\middle|\, \alpha \in I_{R}\right\}$ of \emph{rule instances}, where each $\Phi_\alpha \subseteq \mathcal{F}$ is a finite set of formulas called the \emph{premises} of the particular rule instance, and $\phi_\alpha \in \mathcal{F}$ is the \emph{conclusion} of the rule instance.
\end{definition}

In other words, an inference rule $R$ for a deductive system is a binary relation between $\mathcal{P}_f(\mathcal{F})$ and $\mathcal{F}$.\footnote{An inference rule $R$ is usually required to be a \emph{partial function} from $\mathcal{P}_f(\mathcal{F})$ to $\mathcal{F}$, so that the result of applying $R$ is fully determined by the premises to which it is applied.
However we will not require this.}
We use the usual proof tree notation
\begin{prooftree*}
\hypo{\phi^1_\alpha}
\hypo{\cdots}
\hypo{\phi^n_\alpha}
\infer3[$R$]{\phi_\alpha}
\end{prooftree*}
where $\Phi_\alpha = \{\phi^1_\alpha, \dotsc, \phi^n_\alpha\}$, to denote instances of a rule $R$ (the order of the premises does not matter).
Defining a rule by specifying all its members is frequently infeasible, so we will instead define a rule by specifying a \emph{schema} for its instances, in a similar way to how one defines axiom schemas for ZFC and other theories.
Rule schemas frequently include \emph{side conditions} further constraining the rule instances.

\begin{definition}
An inference rule of the form
\[ R = \left\{ (\emptyset, \phi_\alpha) \,\middle|\, \alpha \in I_R \right\} \]
is known as an \textbf{axiom schema}.
The conclusions $\phi_\alpha$ of its individual instances are known as \textbf{axioms}.
\end{definition}

Let us view the theory $HoTT$ as a deductive system.
Its language consists of the expressions of the typed $\lambda$-calculus, the term constructors given by the type formation and introduction rules, the univalence constructor, the context variables $\Gamma$, $\Delta$, and the improper symbols $\ctx$, $\context$, $\equiv$ and $\colon$.
There are distinguished syntactic categories of \emph{terms} and \emph{types}.
Its formulas are given by the judgment forms described in Section \ref{sec:judgments}, and its inference rules are those described in Sections \ref{sec:ITT} and \ref{sec:HoTT}.
There is only one axiom of $HoTT$, namely the conclusion of $\ctx$-\emp,
\[ \cdot\ \ctx. \]
Note that in the terminology of deductive systems, the univalence axiom schema is really a rule schema.

The following definitions are standard.

\begin{definition} \label{def:derivation}
Let $\mathcal{D} = (\mathcal{L}, \mathcal{F}, \mathcal{R})$ be a deductive system, $\Gamma \subseteq \mathcal{F}$, and $\varphi \in \mathcal{F}$.
A \textbf{derivation of $\varphi$ from premises $\Gamma$ in $\mathcal{D}$} is a labelled tree:
\begin{enumerate}[label=(\roman*)]
\item whose root is labelled $\varphi$,
\item whose set of leaf labels is $\Gamma$, and
\item where every label of an internal node $u$ is the conclusion of an instance of an inference rule whose premises are the labels of the children of $u$.
\end{enumerate}
Leaf labels that are not axioms are the \textbf{assumptions} of the derivation.
A derivation of $\varphi$ only from axioms (i.e.\ without assumptions) is a \textbf{derivation of $\varphi$}.
\end{definition}

We will write derivations using the usual proof tree notation, with conclusion on the bottom.

We write
\[ \Gamma \turnstile_{\mathcal{D}} \varphi \]
to express that there is a derivation of $\varphi$ from premises $\Gamma$ in $\mathcal{D}$.
In such a case $\varphi$ is said to be \emph{derivable} from $\Gamma$.
A derivation of $\varphi$ without assumptions may be written
\[ \turnstile_{\mathcal{D}} \varphi \]
instead.

\begin{definition} \label{def:theorem}
A \textbf{theorem} of a deductive system $\mathcal{D}$ is a formula $\varphi$ of $\mathcal{D}$ such that $\turnstile_{\mathcal{D}} \varphi$, i.e.\ $\varphi$ is derivable from the axioms of $\mathcal{D}$ alone.
\end{definition}

The following concepts will be used when we later come to analyze our implementation of $HoTT$.

\begin{definition}
Let $\mathcal{D} = (\mathcal{L}, \mathcal{F}, \mathcal{R})$ be a deductive system, and let
\[ R = \left\{ (\Phi_\alpha, \phi_\alpha) \,\middle|\, \alpha \in I_R \right\} \]
be a binary relation from $\mathcal{P}_f(\mathcal{F})$ to $\mathcal{F}$, not necessarily in $\mathcal{R}$.
$R$ is called:
\begin{enumerate}[label=(\roman*)]
\item a \textbf{derivable rule in $\mathcal{D}$} if $\Phi_\alpha \vdash_\mathcal{D} \phi_\alpha$ for every instance $(\Phi_\alpha, \phi_\alpha)$ of $R$.

\item an \textbf{admissible rule in $\mathcal{D}$} if adding $R$ to $\mathcal{R}$ does not increase the set of theorems of $\mathcal{D}$, i.e.\ if every theorem of $\mathcal{D} + R$ is a theorem of $\mathcal{D}$.

\item a \textbf{correct rule in $\mathcal{D}$} if for every instance $(\Phi_\alpha, \phi_\alpha)$ of $R$, whenever the formulas in $\Phi_\alpha$ are all theorems of $\mathcal{D}$, then so is $\phi_\alpha$.
\end{enumerate}
\end{definition}

It is clear that a rule $R$ derivable in $\mathcal{D}$ is also admissible in $\mathcal{D}$: consider any derivation that uses $R$ and replace the application of $R$ with its derivation in $\mathcal{D}$.
The converse is not true in general (see Remark 15.2 of \cite{lambda-calculus}).

We also have equivalence of rule admissibility and correctness for deductive systems:

\begin{proposition} \label{prop:admissible-correct}
$R$ is an admissible rule in $\mathcal{D}$ if and only if it is correct in $\mathcal{D}$.
\end{proposition}

\begin{proof}
Assume $R$ is admissible in $\mathcal{D}$ and let
\[
\begin{prooftree}
    \hypo{\phi^1_\alpha}
    \hypo{\cdots}
    \hypo{\phi^n_\alpha}
    \infer3[$R$]{\phi_\alpha}
\end{prooftree}
\]
be an instance of $R$ whose premises are all theorems of $\mathcal{D}$.
Then $\phi_\alpha$ is a theorem of $\mathcal{D} + R$, and by admissibility of $R$ is also a theorem of $\mathcal{D}$.

Conversely let $R$ be correct in $\mathcal{D}$, and let $D$ be a derivation of a theorem $\varphi$ that uses $R$.
For each application
\[
\begin{prooftree}
    \hypo{\phi^1_\alpha}
    \hypo{\cdots}
    \hypo{\phi^n_\alpha}
    \infer3[$R$]{\phi_\alpha}
\end{prooftree}
\]
of an instance of $R$ in $D$, by correctness of $R$ the subderivation of $D$ with root $\phi_\alpha$ may be replaced with a derivation of $\phi_\alpha$ that uses only the rules of $\mathcal{D}$.
Doing this iteratively for every occurrence of an application of $R$ yields a derivation of $\varphi$ using only rules of $\mathcal{D}$, so $\varphi$ is a theorem of $\mathcal{D}$.
(One proves this more rigorously using induction.)
\end{proof}

\subsection{Natural deduction systems}

Here we give a brief indication of how one can define \emph{natural deduction systems}, which differ from Hilbert-style deductive systems in that they incorporate the notions of \emph{assumption} and \emph{discharge}.

\begin{definition}
A \textbf{natural deduction system} is a triple
\[ \mathcal{N} = (\mathcal{L}, \mathcal{F}, \mathcal{R}) \]
whose language $\mathcal{L}$ and formulas $\mathcal{F}$ are defined as for deductive systems.
A rule $R \in \mathcal{R}$ is a set of the form
\[ R = \left\{(\Phi_\alpha, \phi_\alpha) \,\middle|\, \alpha \in I_R \right\} \]
where each $\phi_\alpha \in \mathcal{F}$ and each $\Phi_\alpha$ is a finite set whose elements are now either formulas $\phi^i_\alpha \in \mathcal{F}$ or \emph{pairs} $(\psi^i_\alpha, \phi^i_\alpha) \in \mathcal{F} \times \mathcal{F}$.
As before, the elements $(\Phi_\alpha, \phi_\alpha)$ of $R$ are known as the \emph{instances} of the rule $R$.
\end{definition}

Instances $(\Phi_\alpha, \phi_\alpha) = (\{ (\psi^{i_1}_\alpha, \phi^{i_1}_\alpha), \dotsc, (\psi^{i_m}_\alpha, \phi^{i_m}_\alpha), \phi^{j_1}_\alpha, \dotsc, \phi^{j_n}_\alpha \}, \phi_\alpha)$ of a rule $R$ are written

\[
\begin{prooftree}
    \hypo{[\psi^{i_1}_\alpha]}
    \ellipsis{}{\phi^{i_1}_\alpha}
    \hypo{\cdots}
    \hypo{[\psi^{i_m}_\alpha]}
    \ellipsis{}{\phi^{i_m}_\alpha}
    \hypo{\phi^{j_1}_\alpha}
    \hypo{\cdots}
    \hypo{\phi^{j_n}_\alpha}
    \infer6[$R$]{\phi_\alpha}
\end{prooftree}
\]
\smallskip

That is, the first coordinate of a pair $(\psi^i_\alpha, \phi^i_\alpha)$ in $\Phi_\alpha$ expresses an assumption that is discharged by application of the rule.
Again, the order of the premises does not matter.

The definitions of axiom schema and axiom for deductive systems carry over to the case of natural deduction systems.

We can define the notion of a derivation in a natural deduction system by modifying the definition for deductive systems to incorporate discharge of assumptions, essentially by requiring every internal node $u$ of a proof tree $D$ to be labelled by the conclusion of some rule instance $(\Phi_\alpha, \phi_\alpha)$ where
\begin{enumerate}[label=(\roman*)]
\item every $\phi^i_\alpha \in \Phi_\alpha$ is the label of a child of $u$ in $D$, and
\item for every $(\psi^i_\alpha, \phi^i_\alpha) \in \Phi_\alpha$, there is a child $v$ of $u$ labelled $\phi^i_\alpha$ and a branch ending at $v$ whose leaf is labelled $\psi^i_\alpha$.
\end{enumerate}
As we will not actually use such a definition in this thesis, we will not say any more beyond this.

Note that every natural deduction system $\mathcal{N} = (\mathcal{L}, \mathcal{F}, \mathcal{R})$ has a (possibly trivial) Hilbert-style \textbf{deductive core} $\mathcal{D}_\mathcal{N} = (\mathcal{L}, \mathcal{F}, \mathcal{R}^\prime)$ where $\mathcal{R}^\prime \subseteq \mathcal{R}$ is the collection of rules of $\mathcal{N}$ whose premises are all formulas, i.e.\ they do not involve discharge of assumptions.
In $\mathcal{D}_\mathcal{N}$, all the definitions and results of Section \ref{sec:deductive-systems} hold.
We will formulate the Pure logic as a natural deduction system, but will thereafter mainly work in its deductive core.

\section{The Pure logic}

The logic of Pure \cite{paulson-foundation} is based on Church's simple type theory, also known as higher-order logic \cite{church-simple-type-theory, andrews-math-logic}.
It is called \emph{simple} because types do not depend on any other values, in contrast to dependent type theory, where types may depend on terms, or polymorphic type theory where they may depend on other types (although as we will see, the Pure logic implements type polymorphism in practice).

\subsection{Types and terms} \label{sec:types-and-terms}

The types of Pure are generated inductively by the basic types, and function types of the form
\begin{pure}
T \Map S
\end{pure}
where \code{T} and \code{S} are types.
The single basic type declared by Pure is \code{prop}, the type of logical propositions; additional basic types may be declared by object logics.

The terms of Pure are the usual terms of the typed lambda-calculus, classified into three major syntactic classes: constants, variables, and lambda expressions.
Isabelle annotates terms with the syntactic class to which they belong.
The terms are further augmented with the three constants \All, $\Longrightarrow$, and $\equiv$, which we will define shortly.
Object logics may also define additional constants.
Terms of type \code{prop} are called \emph{formulas}; they are the judgments of the Pure logic.
They are also the formulas of the theory, seen as a natural deduction system.

The application of a term \code{t} to argument terms \code{t\sub{1}}, \textellipsis, \code{t\sub{n}} is written in higher-order style, i.e.\ as
\begin{pure}
(t t\sub{1} \textellipsis t\sub{n})
\end{pure}
instead of \code{t(t\sub{1}, \textellipsis, t\sub{n})}.
In a term \code{t} obtained by the successive application of terms to arguments, the outermost leading term is known as the \emph{head} of \code{t}.

Term expressions in Pure may contain free variables.
As usual we will write
\[
\code{t}[\code{a\sub{1}, \textellipsis, a\sub{n}}/\code{x\sub{1}, \textellipsis, x\sub{n}}]
\]
to denote the simultaneous replacement of the free variables \code{x\sub{i}} with the terms \code{a\sub{i}} in the term \code{t}.\footnote{This is a notation of theory, as opposed to actual Pure syntax.
We will see later that although Pure has a notion of ``free variable'', it has no native way of expressing the substitution of terms in place of these free variables.
We will instead have to encode substitutions \code{t$[$a$/$x$]$} as applications of function terms \code{t} of type \code{T \Map\ S} to terms \code{a}.}

We write
\begin{pure}
t \Dblcolon T
\end{pure}
to express that term \code{t} has type \code{T}.
Note that this expression is not a formula but simply a type annotation, and only appears as a fragment of a larger formula.
Isabelle performs automatic inference of the most general type of entities in the system, so type annotations on terms are optional, except in cases where a more specific type is needed.

Pure's type implementation has parametric polymorphism.
Type variables are prefixed with apostrophes, as in \code{'a}.
Note that the original formulation of simple type theory is non-parametric, replacing every instance of a polymorphic constant \code{t \Dblcolon\ 'a} with a family of constants \code{t\sub{T} \Dblcolon\ T} for every type \code{T}.

\subsection{Logical constants}

Pure provides three constants for constructing formulas,
\begin{align*}
\equiv & \dblcolon \texttt{'a} \Rightarrow \texttt{'a} \Rightarrow \mathtt{prop} \\
\Longrightarrow & \dblcolon \mathtt{prop} \Rightarrow \mathtt{prop} \Rightarrow \mathtt{prop} \\
\bigwedge & \dblcolon \texttt{('a} \Rightarrow \texttt{prop)} \Rightarrow \mathtt{prop}
\end{align*}
expressing definitional equality, implication, and universal quantification respectively.
Note the difference between the Pure arrows \Map\ and \Implies: the first is the function type constructor, which takes type arguments; the second is the implication operator, which takes formula arguments.

We write, as usual, \code{(a \Equiv\ b)} and \code{(p \Implies\ q)} in place of \code{(\Equiv\ a b)} and \code{(\Implies\ p q)}.
Universal quantification is written as
\begin{pure}
\All{}x. p
\end{pure}
instead of \code{\All($\lambda$x. p)}.
By convention, the scope of the binder \All\ extends to the end of the line.
We further abbreviate
\begin{pure}
\All{}x\sups{1}. \Dotsc. \All{}x\sups{k}. p
\end{pure}
by
\begin{pure}
\All{}x\sups{1} \Dotsc x\sups{k}. p\textrm{.}
\end{pure}
\subsection{Logical rules} \label{sec:pure-logical-rules}

Here we state the theory of the inference rules governing the Pure logical constants.
We will elaborate on their implementation and usage in practice later.
In the following statements, \code{p} and \code{q} are terms of type \code{prop}, \code{x} and \code{y} are variables, and \code{a}, \code{b}, \code{c}, \code{f}, \code{g} stand for arbitrary terms of the appropriate type.
\bigskip

\noindent \emph{Implication:}

\[
\begin{prooftree}[center=false]
    \hypo{[\code{p}]}
    \ellipsis{}{\code{q}}
    \infer1[\Implies\I]{\code{p \Implies\ q}}
\end{prooftree}
\qqquad
\begin{prooftree}[center=false]
    \hypo{\code{p \Implies\ q}}
    \hypo{\code{p}}
    \infer2[\Implies\E]{\code{q}}
\end{prooftree}
\]
\smallskip
where \code{p} is discharged from the premises of \Implies\I.
\bigskip

\noindent \emph{Universal quantification:}

\[
\begin{prooftree}[center=false]
    \hypo{\code{p}}
    \infer1[\All\I]{\code{\All x. p}}
\end{prooftree}
\qqquad
\begin{prooftree}[center=false]
    \hypo{\code{\All x. p}}
    \infer1[\All\E]{\code{p}[\code{t}/\code{x}]}
\end{prooftree}
\]
\smallskip
In \All\I, \code{x} must not be free in any of the assumptions of the derivation of \code{p}.
In \All\E, \code{t} is any term of appropriate type.
\bigskip

\noindent\emph{Definitional equality and lambda expressions:}

\[
\begin{prooftree}[center=false]
    \infer0[\Equiv\R]{\code{a \Equiv\ a}}
\end{prooftree}
\qqquad
\begin{prooftree}[center=false]
    \hypo{\code{a \Equiv\ b}}
    \infer1[\Equiv\s]{\code{b \Equiv\ a}}
\end{prooftree}
\qqquad
\begin{prooftree}[center=false]
    \hypo{\code{a \Equiv\ b}}
    \hypo{\code{b \Equiv\ c}}
    \infer2[\Equiv\T]{\code{a \Equiv\ c}}
\end{prooftree}
\]

\[
\begin{prooftree}[center=false]
    \infer0[$\alpha$-\conv]{\code{(\(\lambda\)x. b) \Equiv\ (\(\lambda\)y. b\([\)y\(/\)x\(]\))}}
\end{prooftree}
\qqquad
\begin{prooftree}[center=false]
    \infer0[$\beta$-\conv]{\code{(\(\lambda\)x. b) a \Equiv\ b\([\)a\(/\)x\(]\)}}
\end{prooftree}
\]

\[
\begin{prooftree}
    \infer0[$\eta$-\conv]{\code{(\(\lambda\)x. f x) \Equiv\ f}}
\end{prooftree}
\]
\smallskip
In $\alpha$-\conv, \code{y} must not be free in \code{b}.
In $\eta$-\conv, \code{x} must not be free in \code{f}.
Note that in contrast to $HoTT$, the Pure definitional equality has $\alpha$-conversion built in.
\bigskip

\noindent\emph{Abstraction and combination rules:}

\[
\begin{prooftree}[center=false]
    \hypo{\code{a \Equiv\ b}}
    \infer1[\textsc{abs}]{\code{(\(\lambda\)x. a) \Equiv\ (\(\lambda\)x. b)}}
\end{prooftree}
\qqquad
\begin{prooftree}[center=false]
    \hypo{\code{f \Equiv\ g}}
    \hypo{\code{a \Equiv\ b}}
    \infer2[\textsc{comb}]{\code{f a \Equiv\ g b}}
\end{prooftree}
\]
In \textsc{abs}, \code{x} must not be free in any of the assumptions of the derivation of \code{a \Equiv\ b}.
\bigskip

\noindent\emph{Function extensionality:}

\[
\begin{prooftree}
    \hypo{\code{f x \Equiv\ g x}}
    \infer1[\ext]{\code{f \Equiv\ g}}
\end{prooftree}
\]
where \code{x} must not be free in \code{f}, \code{g}, or the assumptions of the derivation of \code{f x \Equiv\ g x}.

In fact, $\eta$-conversion and function extensionality are equivalent:

\begin{proposition}[Theorem 7.4, Hindley and Seldin \cite{lambda-calculus}] \label{prop:eta-ext-equiv}
Let $Pure_\eta$ be the theory consisting of the rules given above without the extensionality rule \ext, and let $Pure_\ext$ be the theory consisting of the rules given above without $\eta$-conversion $\eta$-\conv.
Then $\eta$-\conv\ is derivable in $Pure_\ext$ and \ext\ is derivable in $Pure_\eta$.
\end{proposition}

\begin{proof}
Assume \code{x} is not free in \code{f}.
Then we have the following derivation:
\smallskip

\[
\begin{prooftree}
    \infer0[$\beta$-\conv]{\code{(\(\lambda\)x. f x) x \Equiv\ f x}}
    \infer1[\ext]{\code{(\(\lambda\)x. f x) \Equiv\ f}}
\end{prooftree}
\]
\smallskip

Conversely, assume \code{x} is not free in \code{f}, \code{g}, or the assumptions of the derivation of
\begin{pure}
f x \Equiv\ g x\textrm{.}
\end{pure}
Then we have

\[
\begin{prooftree}
    \infer0[$\eta$-\conv]{\code{(\(\lambda\)x. f x) \Equiv\ f}}
    \infer1[\Equiv\s]{\code{f \Equiv\ (\(\lambda\)x. f x)}}

    \hypo{\code{f x \Equiv\ g x}}
    \infer1[\textsc{abs}]{\code{(\(\lambda\)x. f x) \Equiv\ (\(\lambda\)x. g x)}}

    \infer2[\Equiv\T]{\code{f \Equiv\ (\(\lambda\)x. g x)}}

    \infer0[$\eta$-\conv]{\code{(\(\lambda\)x. g x) \Equiv\ g}}

    \infer2[\Equiv\T]{\code{f \Equiv\ g}}
\end{prooftree}
\]
\end{proof}

Pure implements \ext\ as a rule derived from $\eta$-\conv.

\section{The Pure system}

We have so far mainly kept the presentation of Pure theoretical, however it is unavoidable that we will eventually have to consider details of implementation.
In this section we begin to explain how the theory of the previous sections is implemented and used in the Isabelle system.
We have attempted, as best as possible, to be clear in our explanations, but there is no substitute for seeing a concrete example for many of the issues discussed here.
Hence it may be helpful to return to this section after having read Chapter \ref{ch:implementation}.

\subsection{Pure and ML}

The Isabelle system is written in the statically-typed programming language ML, so the terms and types of Pure are themselves implemented as typed ML terms.
Hence there is a distinction between the types \code{T} and terms \code{t} of Pure on the object-level and the types \ml{T} and terms \ml{t} of ML on the meta-level, which is separate from the object/meta distinction between object logics and Pure.

\subsection{Axioms and theorems}

Pure terms of type \code{prop} encompass all formulas, including those that are not necessarily true or provable.
To distinguish those formulas that should be considered true judgments, Isabelle defines ML terms of a special type \ml{thm}, which wrap formulas \code{p} with additional information telling the system to consider \code{p} to be proven.
These \emph{ML theorems} are introduced into an Isabelle proof session in one of two ways: either axiomatically with the Pure command \code{axiomatization}, or constructed from other ML theorems using Pure's proof procedures.

Most of the logical rules in the previous section are defined directly as ML theorems in the Isabelle source files \cite{drule.ml, thm.ml}, object logics introduce additional axioms.

\subsection{Free versus schematic variables}

Before explaining how logical rules are encoded and applied in Isabelle, we need to clarify the distinction between two different kinds of variable in the system.

Recall from Section \ref{sec:types-and-terms} that the Pure terms are classified syntactically into constants, lambda expressions, and variables.
In the actual Pure implementation, a further distinction is made between \emph{free} and \emph{schematic} variables.
Logically, these are the same, but internally they are represented differently: free variables are represented by ML terms with head \ml{Free}; schematic variables by ML terms with head \ml{Var}.
Functionally, during the proof process the Isabelle system is allowed to instantiate instances of schematic variables, but not free variables.
That is, schematic variables, written
\begin{pure}
?x \Dblcolon T\textrm{,}
\end{pure}
indicate ``placeholders'' in formulas which may be instantiated by the system with any Pure term of type \code{T}.
This instantiation occurs when the system performs unification.
On the other hand, a free variable \code{x \Dblcolon\ T} indicates an arbitrary but \emph{fixed} term of type \code{T}, which is not modified by unification.

\subsection{Unification}

\begin{definition}
Let \code{t} and \code{s} be Pure terms containing schematic variables \code{?x\sub{1}}, \textellipsis, \code{?x\sub{n}} and \code{?y\sub{1}}, \textellipsis, \code{?y\sub{m}} respectively.
A \textbf{higher-order unifier} of \code{t} and \code{s} is an assigment $\sigma$ of terms to schematic variables
\begin{gather*}
\code{?x\sub{1}} \coloneqq \code{a\sub{1}}, \dotsc, \code{?x\sub{n}} \coloneqq \code{a\sub{n}} \\
\code{?y\sub{1}} \coloneqq \code{b\sub{1}}, \dotsc, \code{?y\sub{m}} \coloneqq \code{b\sub{m}}
\end{gather*}
that respects types, such that
\[
\code{t}[\sigma] \colonequiv \code{t}[\code{a\sub{1}}, \dotsc, \code{a\sub{n}}/\code{?x\sub{1}}, \dotsc, \code{?x\sub{n}}]
\]
and
\[
\code{s}[\sigma] \colonequiv \code{s}[\code{b\sub{1}}, \dotsc, \code{b\sub{m}}/\code{?y\sub{1}}, \dotsc, \code{?y\sub{m}}]
\]
are syntactically identical modulo $\beta$-reduction.
\end{definition}

For example, the terms
\code{f (?x\sub{1} a)}
and
\code{?y\sub{1} (g ?y\sub{2})}
have a unifier
\[ \code{?x\sub{1}} \coloneqq \code{g},\ \code{?y\sub{1}} \coloneqq \code{f},\ \code{?y\sub{2}} \coloneqq \code{a}. \]
They also admit more complicated unifiers such as
\[ \code{?x\sub{1}} \coloneqq \code{(\(\lambda\)x. g b)},\ \code{?y\sub{1}} \coloneqq \code{f},\ \code{?y\sub{2}} \coloneqq \code{b}. \]
A unifier may leave schematic variables in terms uninstantiated, and may also occasionally introduce new schematic variables.

Isabelle performs higher-order unification using a modification of Huet's procedure \cite{huet75}, and may occasionally fail to return a unifier even if one exists.
This incompleteness is due to the higher-order nature of the problem, in particular the pair of terms
\begin{pure}
?f x \textrm{and} ?g x
\end{pure}
admit infinitely many unifiers, e.g. $\code{?f} = \code{?g} \coloneqq \lambda\code{x. ?h}$.
Such a pair is known as a \emph{flex-flex pair}.
If Isabelle encounters a flex-flex pair, it does not make a choice of unifier and instead separately records the constraint to be handled manually.
We will not discuss the details of the unification process any further, interested readers may refer to Sections 3.1 and 7.4 of \cite{old-intro} for more information.

\section{Reasoning and proof in Pure}

In this section we continue considering implementation details specifically pertaining to the proof process in Isabelle.

\subsection{Inference rules}

Much of Pure's power and flexibility comes from its ability to encode natural deduction rules as formulas.
Hence it can implement the inference rules of many object logics.

The basic pattern is the following.
The Pure implication \Implies\ expresses the notion of \emph{entailment} or \emph{derivability}.
An inference rule schema

\[
\begin{prooftree}
    \hypo{[\psi^{i_1}]}
    \ellipsis{}{\phi^{i_1}}
    \hypo{\cdots}
    \hypo{[\psi^{i_m}]}
    \ellipsis{}{\phi^{i_m}}
    \hypo{\phi^{j_1}}
    \hypo{\cdots}
    \hypo{\phi^{j_n}}
    \infer6[$R$]{\phi}
\end{prooftree}
\]
can be encoded by the formula
\begin{multline} \label{eq:pure-rule-pattern}
\code{
\llb{}\All{}x\subsup{i\sub{1}}{1} \Dotsc\ x\subsup{i\sub{1}}{k\sub{1}}. q\sub{i\sub{1}} \Implies\ p\sub{i\sub{1}};
\Dotsc;
\All{}x\subsup{i\sub{m}}{1} \Dotsc\ x\subsup{i\sub{m}}{k\sub{m}}. q\sub{i\sub{m}} \Implies\ p\sub{i\sub{m}};} \\
\code{\All{}x\subsup{j\sub{1}}{1} \Dotsc\ x\subsup{j\sub{1}}{l\sub{1}}. p\sub{j\sub{1}}; \Dotsc; \All{}x\subsup{j\sub{n}}{1} \Dotsc\ x\subsup{j\sub{n}}{l\sub{n}}. p\sub{j\sub{n}}\rrb\ \Implies\ p},
\end{multline}
where \code{p} is the encoding of $\phi$, the \code{q\sub{\(\ast\)}} and \code{p\sub{\(\ast\)}} are encodings of the respective formula schemas $\psi^\ast$ and $\phi^\ast$ as Pure formulas, and each universal quantifier \All\ optionally binds occurrences of free variables \code{x\subsup{\(\ast\)}{\(\ast\)}} in its scope.
Here we have used the Pure notation
\begin{pure}
\llb{}p\sub{1}; \Dotsc; p\sub{n}\rrb \Implies p
\end{pure}
which abbreviates
\begin{pure}
p\sub{1} \Implies \Dotsb \Implies p\sub{n} \Implies p\textrm{,}
\end{pure}
where \Implies\ is right-associative.

That is, a formula such as \eqref{eq:pure-rule-pattern} encodes an inference rule in which the rightmost implication represents the final ``line'' of the rule, and the implications inside the double brackets \code{\llb\ \rrb} represent subderivations, whose premises are to be discharged.
This structure may furthermore be nested.
In the jargon of proof theory, we say that an inference rule is represented by a Horn clause of hereditary Harrop formulas in Pure.

Inference rules having no premises simply have their conclusion encoded as axioms, so

\[
\begin{prooftree}
    \infer0[$R$]{\phi}
\end{prooftree}
\]
is simply given by an encoding \code{p} of $\phi$.

From now on we will usually use the term ``inference rule'' to refer both to abstract inference rules as defined in Definition \ref{def:deductive-system}, as well as to their encoding as Pure formulas, and trust the context to disambiguate which sense we mean.
In cases where the distinction matters we will be very explicit.

\subsection{Goals and backward proof}

Consider a deductive system $\mathcal{D}$.
In a derivation $D$ of $\varphi$ from $\Gamma$ in $\mathcal{D}$, one starts with the formulas of $\Gamma$ and iteratively applies instances of appropriate inference rules in order to finally obtain $\varphi$, in this way constructing a proof tree from the top down.
In Isabelle, this is known as \emph{forward proof}, where one works forward from assumptions, applying inference rule transformations to finally obtain the conclusion.

We can also read $D$ from the bottom up, starting with the root $\varphi$ which is to be shown.
Now we read the application of inference rule instances \emph{backwards}: a derivation step

\[
\begin{prooftree}
    \hypo{\phi^1}
    \hypo{\cdots}
    \hypo{\phi^n}
    \infer3[$R$]{\phi}
\end{prooftree}
\]
in $D$ expresses that in order to derive the \emph{goal} $\phi$, it suffices to derive the \emph{subgoals} $\phi^1, \dotsc, \phi^n$.
This forms the basis of Isabelle's notion of \emph{backward proof}, which we now illustrate with an example.

Consider the derivation given in the proof of Proposition \ref{prop:eta-ext-equiv},

\[
\begin{prooftree}
    \infer0[$\eta$-\conv]{\code{(\(\lambda\)x. f x) \Equiv\ f}}
    \infer1[\Equiv\s]{\code{f \Equiv\ (\(\lambda\)x. f x)}}

    \hypo{\code{f x \Equiv\ g x}}
    \infer1[\textsc{abs}]{\code{(\(\lambda\)x. f x) \Equiv\ (\(\lambda\)x. g x)}}

    \infer2[\Equiv\T]{\code{f \Equiv\ (\(\lambda\)x. g x)}}

    \infer0[$\eta$-\conv]{\code{(\(\lambda\)x. g x) \Equiv\ g}}

    \infer2[\Equiv\T]{\code{f \Equiv\ g}}
\end{prooftree}
\]
where the aim is to show that $\code{(f x \Equiv\ g x)} \turnstile_{Pure_\eta} \code{(f \Equiv\ g)}$.
Starting from the goal
\begin{equation} \label{eq:goal}
\code{f \Equiv\ g}
\end{equation}
and applying an appropriate instance of the rule \Equiv\T, we obtain the two subgoals
\begin{gather}
\code{f \Equiv\ (\(\lambda\)x. f x)}, \label{eq:subgoal1} \\
\code{(\(\lambda\)x. g x) \Equiv\ g}, \label{eq:subgoal2}
\end{gather}
which we can consider as replacing the original goal, in the sense that if we can derive both \eqref{eq:subgoal1} and \eqref{eq:subgoal2}, then we can also derive \eqref{eq:goal}.

We then repeat this process for each of these subgoals.
The second subgoal \code{(\(\lambda\)x. g x) \Equiv\ g} is the conclusion of an instance of the rule $\eta$-\conv, which has no premises.
Hence applying $\eta$-\conv\ yields no new subgoals in place of \eqref{eq:subgoal2}, that is to say, in order to derive \eqref{eq:subgoal2} it suffices to derive the empty set.
Thus \eqref{eq:subgoal2} is trivially proved, befitting its status as an axiom.

This leaves subgoal \eqref{eq:subgoal1}.
Again we apply an appropriate instance of \Equiv\T\ to replace \eqref{eq:subgoal1} by new subgoals
\begin{gather}
\code{f \Equiv\ (\(\lambda\)x. f x)}, \label{eq:subgoal3} \\
\code{(\(\lambda\)x. f x) \Equiv\ (\(\lambda\)x. g x)}, \label{eq:subgoal4}
\end{gather}
the first of which is then transformed to the trivial goal set by applying \Equiv\s\ and then $\eta$-\conv.
To the second new subgoal \eqref{eq:subgoal4} we apply rule \textsc{abs}, yielding the subgoal
\begin{pure}
f x \Equiv g x\textrm{.}
\end{pure}
This is the single assumption of the derivation we aim to show, and hence we are done.

\begin{remark}
The example above illustrates a backward proof that
\[ \code{(f x \Equiv\ g x)} \turnstile_{Pure_\eta} \code{(f \Equiv\ g)}, \]
subject to the relevant side conditions.
In practice, we would normally perform an initial application of the Pure rule \Implies\I\ in order to show, equivalently, that

\begin{pure}
\(\turnstile\)\sub{\(Pure\)\sub{\(\eta\)}} (f x \Equiv g x) \Implies (f \Equiv g)\textrm{.}
\end{pure}
\end{remark}

\subsection{Proof states, lifting, and resolution}

Backward proof is a basic way Pure constructs new ML theorems from existing ones.
We now describe its implementation.

\begin{definition}
A \textbf{proof state} is a Pure formula of the form
\begin{equation} \label{eq:proof-state}
\code{
\llb{}p\sub{1}; \Dotsc; p\sub{n}\rrb\ \Implies\ p
}
\end{equation}
where the \code{p\sub{i}} are formulas known as the \textbf{subgoals} of the proof state, and \code{p} is a formula known as the \textbf{goal}.
\end{definition}

A proof state \eqref{eq:proof-state} expresses that in order to prove the goal \code{p}, it suffices to prove the subgoals \code{p\sub{1}}, \textellipsis, \code{p\sub{n}}, where \code{n} may be zero.
The goal \code{p} may itself be an implication.
Note that proof states and inference rules have the same syntactic form and dual interpretations.

Proof states are transformed using inference rules by a method called \emph{lifted resolution}.
In the following definitions, let
% \begin{pure}
% P: \llb{}p\sub{1}; \Dotsc; p\sub{n}\rrb \Implies p
% \end{pure}
% be a proof state, and
\begin{pure}
R: \llb{}r\sub{1}; \Dotsc; r\sub{m}\rrb \Implies r
\end{pure}
be an inference rule encoded in the form \eqref{eq:pure-rule-pattern}.

\begin{definition}
Let $\code{q\sub{1}}, \dotsc, \code{q\sub{k}}$ be formulas.
The \textbf{lift of \code{R} over formulas \code{q\sub{1}}, \textellipsis, \code{q\sub{k}}} is the formula
\begin{multline} \label{rule:lift-formulas}
\code{
\llb\llb{}q\sub{1}; \Dotsc; q\sub{k}\rrb\ \Implies\ r\sub{1}; \Dotsc; \llb{}q\sub{1}; \Dotsc; q\sub{k}\rrb\ \Implies\ r\sub{m}\rrb\ \Implies} \\
\code{(\llb{}q\sub{1}; \Dotsc; q\sub{k}\rrb\ \Implies\ r)\textrm{.}
} \tag{\textsc{lift\sub{\textsc{forms}}}(\texttt{R})}
\end{multline}
\end{definition}

It is easily shown that \textsc{lift\sub{\textsc{forms}}}(\code{R}) is in fact a derivable inference rule in the theory $Pure + \code{R}$ consisting of the inference rules of Section \ref{sec:pure-logical-rules} together with \code{R}.

\begin{definition}
Let \code{x\sups{1}}, \textellipsis, \code{x\sups{k}} be terms.
The \textbf{lift of \code{R} over parameters \code{x\sups{1}}, \textellipsis, \code{x\sups{k}}} is the formula
\begin{equation} \label{rule:lift-parameters}
\code{
\llb{}\All{}x\sups{1} \Dotsc\ x\sups{k}.\ r\subsup{1}{\(\ast\)};
\Dotsc;
\All{}x\sups{1} \Dotsc\ x\sups{k}.\ r\subsup{m}{\(\ast\)}\rrb\ 
\Implies\ 
\All{}x\sups{1} \Dotsc\ x\sups{k}.\ r\sups{\(\ast\)}
}, \tag{\textsc{lift\sub{\textsc{params}}}(\texttt{R})}
\end{equation}
where \code{s\sups{$\ast$}} is the formula obtained from \code{s} by replacing every instance of a schematic variable \code{?y} in \code{s} by the schematic term
\begin{pure}
(?y x\sups{1} \Dotsc x\sups{k})\textrm{.}
\end{pure}
\end{definition}

Again, one can show (Section 6.1, \cite{paulson-foundation}) that \textsc{lift\sub{\textsc{params}}}(\texttt{R}) is a derivable rule in $Pure + \code{R}$.

As an illustration consider the following encoding of the $\vee$-introduction rule of intuitionistic propositional logic
\begin{pure}
Or_intro: ?p \Implies (?p \(\vee\) ?q)\textrm{,}
\end{pure}
which corresponds to
\[
\begin{prooftree}
    \hypo{p}
    \infer1[$\vee$\I]{p \vee q}
\end{prooftree}.
\]
Lifting \code{Or\_intro} over the formula \code{r} and the parameter \code{x} yields
\begin{pure}
(\All{}x. r \Implies (?p x)) \Implies (\All{}x. r \Implies ((?p x) \(\vee\) (?q x)))\textrm{,}
\end{pure}
which expresses the derivable inference rule
\[
\begin{prooftree}
    \hypo{[r]}
    \ellipsis{}{p(x)}
    \infer1[]{r \implies (p(x) \vee q(x))}
\end{prooftree}
\]
of intuitionistic predicate logic, where $x$ is arbitrary.
Note that we lifted over the \emph{free} formula \code{r} and not the \emph{schematic} formula \code{?r}, so we do not replace \code{r} by \code{(r x)}.

The derivability of the two lifting rules means that lifting is sound:

\begin{theorem} \label{thm:lifting-sound}
The lift of a Pure theorem over arbitrary formulas and parameters is a Pure theorem.
Applying a lifted Pure theorem to Pure theorems yields a Pure theorem.
\end{theorem}

\begin{proof}
Let \code{R} be a Pure theorem.
Then the lifting rules \textsc{lift\sub{\textsc{forms}}}(\code{R}) and \textsc{lift\sub{\textsc{params}}}(\code{R}) are derivable as abstract inference rules in $Pure + \code{R} = Pure$.
By discharging the premises of their derivations using \Implies\I, we see that they are in fact derivable as concrete theorems of Pure.
Hence the lift of a Pure theorem over arbitrary formulas and parameters is also a Pure theorem.

The lift of a Pure theorem is in fact derivable as an abstract inference rule in the deductive core $\mathcal{D}_{Pure}$ of Pure, i.e.\ it is derivable using only the inference rules that do not discharge assumptions.
Hence it is admissible, and by Proposition \ref{prop:admissible-correct} it is also correct.
Thus applying it to theorems yields a theorem.
\end{proof}

The Pure logical rules and any additional inference rules defined by object logics are theorems of Pure.
Theorem \ref{thm:lifting-sound} then says that applying lifted inference rules to Pure theorems yields a Pure theorem.
To implement backward proof we need one final definition:

\begin{definition}
Let
\begin{equation} \label{eq:proof-state}
\code{
P: \llb{}p\sub{1}; \Dotsc; p\sub{n}\rrb\ \Implies\ p
}
\end{equation}
be a proof state, and
\begin{equation} \label{eq:inf-rule}
\code{
R: \llb{}r\sub{1}; \Dotsc; r\sub{m}\rrb\ \Implies\ r
}
\end{equation}
an inference rule.

Suppose there is a subgoal \code{p\sub{i}} of \code{P}, a lift
\begin{equation} \label{eq:lifted-rule}
\code{
R\sups{\(\uparrow\)}: \llb{}r\subsup{1}{\(\uparrow\)}; \Dotsc; r\subsup{m}{\(\uparrow\)}\rrb\ \Implies\ r\sups{\(\uparrow\)}
}
\end{equation}
of \code{R} over formulas \code{q\sub{1}}, \textellipsis, \code{q\sub{k}} and parameters \code{x\sups{1}}, \textellipsis, \code{x\sups{l}}, and a unifier $\sigma$ of \code{r\sups{\(\uparrow\)}} and \code{p\sub{i}}.
Then the \textbf{resolvent of \code{P} with \code{R} with respect to subgoal \code{p\sub{i}}} is the formula
\begin{equation} \label{eq:resolvent}
\code{
(\llb{}p\sub{1}; \Dotsc; p\sub{i-1}; r\subsup{1}{\(\uparrow\)}; \Dotsc; r\subsup{m}{\(\uparrow\)}; p\sub{i+1}; \Dotsc; p\sub{n}\rrb\ \Implies\ p)\([\sigma]\)
}
\end{equation}
\end{definition}

The resolvent of \code{P} with \code{R} is a new proof state.
In words, it is formed by finding a subgoal \code{p\sub{i}} with which the lifted conclusion of \code{R} unifies, and then replacing \code{p\sub{i}} with new subgoals given by the lifted hypotheses of \code{R} and propagating the assignments given by the unifier throughout the new proof state.

Given a rule \code{R} and a proof state \code{P} as in \eqref{eq:proof-state} and \eqref{eq:inf-rule}, the process of finding a resolvent of \code{P} with \code{R} is known as \textbf{resolution}.
It can be encapsulated by the following rule, where \code{P\sub{RES}} is as given in \eqref{eq:resolvent}:
\[
\begin{prooftree}
    \hypo{\code{R}}
    \hypo{\code{P}}
    \infer2[\res]{\code{P\sub{RES}}}
\end{prooftree}
\]
which holds whenever there is a unifier $\sigma$ of some subgoal \code{p\sub{i}} of \code{P} with the conclusion of \code{R\sups{\(\uparrow\)}} lifted over the premises and parameters of \code{p\sub{i}}.

\begin{theorem}
The inference rule \res\ is correct in Pure.
\end{theorem}

\begin{proof}
Assume that \code{R} and \code{P} are theorems of Pure.
By Theorem \ref{thm:lifting-sound}, \code{R\sups{$\uparrow$}} is a theorem.
Also, \code{R\sups{$\uparrow$}}$[\sigma]$ and \code{P}$[\sigma]$ are theorems for any instantiation $\sigma$ of schematic variables.

Consider \code{P\sub{RES}} as an abstract inference rule
\[
\begin{prooftree}
    \hypo{\code{p\sub{1}}[\sigma]}
    \hypo{\dotsb}
    \hypo{\code{p\sub{i-1}}[\sigma]}
    \hypo{\code{r\subsup{1}{\(\uparrow\)}}[\sigma]}
    \hypo{\dotsb}
    \hypo{\code{r\subsup{m}{\(\uparrow\)}}[\sigma]}
    \hypo{\code{p\sub{i+1}}[\sigma]}
    \hypo{\dotsb}
    \hypo{\code{p\sub{n}}[\sigma]}
    \infer9[]{\code{p}[\sigma]}
\end{prooftree}.
\]
This rule is derivable in $Pure + \code{R\sups{$\uparrow$}}[\sigma] + \code{P}[\sigma]$ since from \code{r\subsup{1}{\(\uparrow\)}}$[\sigma]$, \textellipsis, \code{r\subsup{m}{\(\uparrow\)}}$[\sigma]$ we can derive \code{r\sups{$\uparrow$}}$[\sigma]$, which is syntactically identical to \code{p\sub{i}}$[\sigma]$.
In turn, together with the other premises \code{p\sub{j}}$[\sigma]$ we can then derive \code{p}$[\sigma]$.
Finally, by applying \Implies\I\ we obtain a derivation
\[ \turnstile_{Pure + \code{R\sups{$\uparrow$}}[\sigma] + \code{P}[\sigma]} \code{P\sub{RES}}. \]
But $Pure + \code{R\sups{$\uparrow$}}[\sigma] + \code{P}[\sigma] = Pure$ since \code{R\sups{$\uparrow$}}$[\sigma]$ and \code{P}$[\sigma]$ are theorems.
Hence \code{P\sub{RES}} is a Pure theorem.
\end{proof}

\subsection{Proof by resolution}

The results of the last section give us a rigorous theoretical basis for backward proof in Pure.

Assume we are working in an extension
\[ Pure + \mathcal{R} \]
of $Pure$, with additional inference rules $\mathcal{R}$ encoded as Pure formulas by some object logic.
To prove a formula \code{p}, we begin with the proof state
\begin{pure}
p \Implies p\textrm{,}
\end{pure}
which is trivially a Pure theorem.
We transform the proof state by applying instances of \res\ with appropriate inference rules: this top-down proof using resolution corresponds to building a proof tree using the inference rules from the bottom up.
Searching for a unifier in each resolution step corresponds to choosing a suitable inference rule instance to apply.

If an instance of an axiom \code{R} occurs as a subgoal \code{p\sub{i}} of a proof state, unifying with \code{R} with respect to \code{p\sub{i}} will result in \code{p\sub{i}} being replaced with no new subgoals, effectively removing a subgoal from the proof state.
Our aim is thus to transform a proof state in order to eliminate subgoals in this way.
Once we have succeeded in removing all subgoals, the proof state will read
\begin{pure}
\Implies p\([\sigma]\)
\end{pure}
or trivially, \code{p}$[\sigma]$, where there may be instantiations of schematic variables in \code{p} from the unification performed during resolution.
From the results of the previous section, we then have that \code{p}$[\sigma]$ is a theorem of $Pure + \mathcal{R}$.

\subsection{Proof by assumption}

Resolution is one of Isabelle's many methods of proof, which are known as \emph{tactics}.
We briefly mention another tactic here, called \emph{assumption}.
It is trivial but often necessary, and proves formulas of the form
\begin{pure}
\llb{}p\sub{1}; \Dotsc; p\sub{n}\rrb \Implies p
\end{pure}
where one of the premises \code{p\sub{i}} unifies with the conclusion \code{p}.
It is an easy exercise to show that assumption can be formulated as a derivable inference rule in Pure.

There are numerous examples of lifting and proof by resolution and assumption in Sections 5.3 and 5.4 of \cite{old-intro}, as well as in \cite{paulson-foundation}.

\chapter{Implementing $HoTT$ in Pure} \label{ch:implementation}

Now that we have a basic understanding of the Isabelle/Pure framework, we are ready to implement $HoTT$ as an object logic inside Pure.
We will call the theory of our implementation $HoTT_P$.\footnote{Other names that were considered but rejected were $PHo$ (for being too culinary) and $PHTT$  (too dismissive).
A personal favorite of the author's was $HoTTiP$, aka ``\textit{Ho}motopy \textit{T}ype \textit{T}heory \textit{i}n \textit{P}ure'', which in his admittedly biased opinion reflects the promise of this area of research.}
Currently the implementation itself is provided as an Isabelle object logic library \emph{Isabelle/HoTT} \cite{isabelle-hott}.
In this chapter we provide a high-level overview of the theory and implementation of $HoTT_P$, leaving the details to the description text contained within the library files.

\section{Foundational setup}

\subsection{Object types and terms}

We begin by defining a new basic Pure type
\begin{pure}
Term\textrm{,}
\end{pure}
whose terms will be the object-level types and terms of $HoTT_P$.
In formulations of type theory without universes or with Tarski-style universes, one would usually define distinct metatypes \code{i} and \code{t} for the object terms and object types respectively.
This is done for example in $CTT$, the Pure implementation of Martin-L\"{o}f constructive type theory.
However, when we use Russell-style universes, the universe types may act as terms in some judgments and as types in others.
This forces us to treat object terms and types uniformly, as instances of the same sort of metaterm.

\begin{remark}
One can view this setup as a triply-sorted pure type system (PTS) \cite{nlab-pts}, where the sorts are the sort of object-level terms $term$, the sort of object-level types $type$, and \code{Term}.
The axioms are then
\[ term \colon \code{Term} \quad \text{and} \quad type \colon \code{Term}. \]
We will however not use the PTS framework in this thesis.
\end{remark}

\subsection{Judgments} \label{sec:judgment-forms}

Next we need to implement the judgments of $HoTT$.
Judgments of object logics are implemented in Pure by defining a function term that takes arguments of judgment forms and wraps them to return a Pure formula.
In an Isabelle theory session we define a new Pure constant
\begin{pure}
has_type \Dblcolon [Term, Term] \Map prop
\end{pure}
that takes two arguments of type \code{Term} and returns a Pure formula.
In this way, formulas
\begin{pure}
(has_type a A)
\end{pure}
now encode the $HoTT$ typing assignment $a \colon A$ and are able to be reasoned about using the proof methods discussed in Chapter \ref{ch:2}.
Instead of \code{(has\_type a A)}, we will write
\begin{pure}
a\Colon{}A\textrm{.}
\end{pure}
Note that the notation
\begin{pure}
[t\sub{1}, \Dotsc, t\sub{n}] \Map t
\end{pure}
abbreviates
\begin{pure}
t\sub{1} \Map \Dotsc \Map t\sub{n} \Map t\textrm{,}
\end{pure}
where \Map\ is right-associative.

Recall that there are three judgment forms of $HoTT$,
\[ \Gamma\ \ctx, \quad \Gamma \turnstile a \colon A, \quad \Gamma \turnstile a \equiv b \colon A, \]
all of which involve contexts $\Gamma$.
It does not appear to be straightforward to directly implement contexts in the Pure framework---a particular difficulty is the question of how variables should be assigned object-level types.
Hence we handle variable typing assignments differently in Pure, and encode the $HoTT$ judgments
\[ x_1 \colon A_1, \dotsc, x_n \colon A_n \turnstile a \colon A \]
and
\[ x_1 \colon A_1, \dotsc, x_n \colon A_n \turnstile a \equiv b \colon A \]
as the $HoTT_P$ formulas
\begin{multline} \label{eq:hottp-type-judgment}
\code{
\All{}x\sub{1} \Dotsc\ x\sub{n}. \llb{}x\sub{1}\Colon{}A\sub{1}; x\sub{2}\Colon{}(A\sub{1} x\sub{1}); \Dotsc; x\sub{n}\Colon{}(A\sub{n} x\sub{1} \Dotsc\ x\sub{n-1})\rrb\ \Implies
} \\
\code{
(a x\sub{1} \Dotsc\ x\sub{n})\Colon{}(A x\sub{1} \Dotsc\ x\sub{n})
}
\end{multline}
and
\begin{multline} \label{eq:hottp-eq-judgment}
\code{
\All{}x\sub{1} \Dotsc\ x\sub{n}. \llb{}x\sub{1}\Colon{}A\sub{1}; x\sub{2}\Colon{}(A\sub{1} x\sub{1}); \Dotsc; x\sub{n}\Colon{}(A\sub{n} x\sub{1} \Dotsc\ x\sub{n-1})\rrb\ \Implies
} \\
\code{
(a x\sub{1} \Dotsc\ x\sub{n}) \Equiv\ (b x\sub{1} \Dotsc\ x\sub{n})
}
\end{multline}
respectively.
Observe that whereas in $HoTT$ we treat the terms $a$, $A$, etc. as \emph{term expressions} which may contain \emph{free variables} $x_i$, in $HoTT_P$ we must consider them as \emph{function terms} \code{a}, \code{A}, etc. of the Pure framework, and explicitly express their dependency on the \emph{parameters} \code{x\sub{i}} (though this dependency is optional).

There is no equivalent in $HoTT_P$ of the context judgment
\[ \Gamma\ \ctx \]
of $HoTT$.
Its purpose in $HoTT$ is to ensure well-typedness of terms and well-formedness of dependent types when formulating the type formation and introduction rules; we will later accomplish the same effect using universal quantification and premises giving type constraints for parameters.
In some cases, this means we have to add premises to the type rules in $HoTT_P$ in order to maintain certain properties present in $HoTT$ that would be lost if we simply translated the rules ``na\"{i}vely''.

\subsection{Definitional equality}

We remark briefly on the implementation of definitional equality.

We use the Pure equality as the definitional equality in $HoTT_P$, and do not define a separate judgment function for it as we did for the typing judgment.
Definitional equality is typed in $HoTT$, but the Pure equality is untyped, and when we translate the definitional equality rules over into $HoTT_P$ as in \eqref{eq:hottp-eq-judgment} we simply drop the type annotation.
This means that in $HoTT_P$ it automatically holds that \code{a\Colon{}A} and \code{a \Equiv\ b} together imply \code{b\Colon{}A}, hence subject reduction automatically holds.

It turns out that from the Pure rules we can derive most of the rules governing definitional equality in $HoTT_P$, and we only have to explicitly axiomatize term congruence rules for some of the constructors for the dependent types.

An issue to be aware of is that the definitional equality of $HoTT$ does not have $\alpha$-conversion, which is instead handled on the meta-level.
However, the Pure equality does have rules for $\alpha$-conversion, thus our implementation conflates this aspect of the object- and meta-theory of $HoTT$.
Functionally, however, there is no difference.
Such issues are known and solutions have been considered before, we refer to \cite{tasistro93} for further reading.

\subsection{Universes}

To implement the universe hierarchy we first declare another basic Pure type \code{Ord} representing meta-level natural numbers, along with new defined constants
\begin{pure}
O \Dblcolon Ord\textrm{,}
S \Dblcolon Ord \Map Ord
\end{pure}
and the strict and non-strict order relations
\begin{pure}
\Lt \Dblcolon [Ord, Ord] \Map prop\textrm{,}
\Leq \Dblcolon [Ord, Ord] \Map prop\textrm{,}
\end{pure}
satisfying the axioms
\begin{equation} \label{eq:Ord-rules}
\begin{gathered}
\code{\All{}n. O \Lt\ (S n)\textrm{,}} \\
\code{\All{}n. O \Leq\ n\textrm{,}} \\
\code{\All{}n. n \Lt\ (S n)}\textrm{,} \\
\code{\All{}n. n \Leq\ (S n)}\textrm{,} \\
\code{\All{}m n. m \Lt\ n \Implies\ (S m) \Lt\ (S n)\textrm{,}} \\
\code{\All{}m n. m \Leq\ n \Implies\ (S m) \Leq\ (S n)\textrm{.}}
\end{gathered} \tag{\textit{Ord}}
\end{equation}
This is all we need to define the terms of universe hierarchy as values of the function term
\begin{pure}
U \Dblcolon Ord \Map Term\textrm{.}
\end{pure}
We give the rules governing the universes in the next section.

\section{Design issues}

\subsection{Analytic versus synthetic formulations}

We define the object-level types and terms as terms of the Pure framework.
In particular, the object-level typing judgment is implemented as a Pure term, and hence the entirety of $HoTT_P$ is contained ``separately'' inside the Pure metalogic.

This is in contrast to systems like Coq and Agda, there the types of $HoTT$ may be \emph{directly} encoded as types of the meta-framework.
This is because those systems are themselves based on extensions of Martin-L\"{o}f type theory, whereas the simple type theory that Pure is based on does not have dependent types and hence cannot natively express the $HoTT$ types.

This illustrates the distinction between the \emph{synthetic} and \emph{analytic} formulations \cite{nlab-lf, harper12} of a logic $\mathcal{L}$ inside a logical framework $\mathcal{F}$.
In the analytic approach we interpret the types and judgments of $\mathcal{L}$ directly using the semantics of $\mathcal{F}$, which usually entails using the types of $\mathcal{F}$ as types of $\mathcal{L}$.
This approach, however, requires the logics $\mathcal{L}$ and $\mathcal{F}$ to be similar.
In the synthetic approach, one simply uses the syntax of $\mathcal{F}$ to encode the entities of $\mathcal{L}$, as we do for $HoTT$ inside Pure, essentially because we have no other choice.

\subsection{Monomorphism versus polymorphism}

The theory of $HoTT$ is \emph{monomorphic with type annotations}, meaning that terms are decorated with the types of their arguments, and are not polymorphic---each of their arguments can only come from one specific type (these are sometimes called \emph{Church-style} terms).
For example, we annotate lambda terms $\lambda(x \colon A).\,x$ with the type $A$ of their argument, and dependent eliminators with the type family $C$ over which they are constructed, e.g. the first argument of
\[ \ind_{\sum}(z.C, x.y.f, p) \colon C[p/z]. \]

Monomorphic theories with annotations have several advantages, first of which is that different terms are distinguished purely syntactically, so for example the lambda terms $\lambda(x \colon A).\,x$ and $\lambda(x \colon B).\,x$ are different identity functions for $A \not\equiv B$.
As another example, the projection functions for dependent pairs $\sum_{x \colon A} B$ are defined separately for each instance of $A$ and $B$, so that we do not just write $\fst(p)$ but
\[ \fst_{A,B}(p) \colonequiv \ind_{\sum}(z.A, x.y.x, p) \colon A, \]
one distinct constant for every $A$ and $B$.

The second advantage of annotating terms $t$ is that the type annotations contain all the information required to reconstruct the derivation of $t$.
Furthermore, such annotations are usually necessary for a theory to have decidable type inference; type theories without annotations are known to suffer from undecidable typability (see for example \cite{dowek93}).

Unfortunately, the monomorphic theory has some distinct \emph{disadvantages} for us when we try to implement it in Pure.
Firstly---and this is true of any implementation of a type theory in general---the annotations create extra notational baggage that must be carried around during the course of a proof, and which is yet irrelevant for most of it.
For example, in a formal proof using path induction, it suffices to know the form of the type family $C$ over which we are inducting at the moment we apply the equality elimination rule.
Once we have done so and instantiated any relevant variables, the proof continues without having to explicitly use the form of $C$ again.

Secondly, recording $C$ as part of the proof term we are constructing creates an additional component of the expression that must be checked for well-formedness further on in the proof, often multiple times.
This second issue is a particular problem for us when we work in Isabelle, due to the fact that in the Pure framework we cannot easily implement automatic type checking or inference, which is normally how the overhead is mitigated.

Thus the monomorphic version of the theory is useful if we want to be able to normalize terms or check for well-formedness of proofs, but it is less convenient if our main purpose is to derive proof terms of given types.
Hence we choose to implement $HoTT_P$ as a \emph{polymorphic} theory without annotations.
Here, for example, the object lambda expression \code{\blm{}x. x} has a different semantic interpretation depending on whether it is considered an inhabitant of
\begin{pure}
(U O) \rarrow (U O)
\end{pure}
or
\begin{pure}
Nat \rarrow Nat\textrm{,}
\end{pure}
and in fact in our theory we will be able to prove for any type \code{A} that
\begin{pure}
\blm{}x. x\Colon\,A \rarrow A\textrm{.}
\end{pure}
We also drop the type family annotations on the dependent eliminators.
Terms in this formulation are sometimes called \emph{Curry-style} terms.
Formulating $HoTT_P$ as a polymorphic theory does not affect the form of the type rules other than to drop the type annotations.

\section{Logical rules} \label{sec:HoTTP-rules}

We formalize the rules in Chapter \ref{ch:1} inside Pure according to the patterns indicated in Section \ref{sec:judgment-forms}.
We present inference rules abstractly, but in the actual implementation these are encoded according to the pattern \eqref{eq:pure-rule-pattern} discussed earlier.

\subsection{Structural rules}

As our theory does not have contexts, there are no equivalents of $\ctx$-\emp, $\ctx$-\ext, or \vbl\ in $HoTT_P$.
As the substitution and weakening rules are admissible (Section A.2.2, \cite{hottbook}), we also do not need to axiomatize them.
The $HoTT$ rules \Equiv-\textsc{refl}, \Equiv-\textsc{sym}, and \Equiv-\textsc{trans} are effectively implemented as the Pure rules \Equiv\I, \Equiv\s, and \Equiv\T.
\textsc{type}-\cng\ is derivable using the Pure rule \textsc{comb}, and \textsc{def}-\cng\ becomes trivial in $HoTT_P$.

\subsection{Universes}

The universe rules are formulated differently from those of $HoTT$ in order to make proofs in Isabelle easier.
They are

\[
\begin{prooftree}[center=false]
    \hypo{\code{i \Lt\ j}}
    \infer1[\code{U\_hierarchy}]{\code{(U i)\Colon\,(U j)}}
\end{prooftree}
\qqquad
\begin{prooftree}[center=false]
    \hypo{\code{A\Colon\,(U i)}}
    \hypo{\code{i \Leq\ j}}
    \infer2[\code{U\_cumulative}]{\code{A\Colon\,(U j)}}
\end{prooftree}
\]
where Isabelle/Pure's automatic meta-level type-checking ensures that \code{i} and \code{j} have Pure type \code{Ord} and \code{A} has Pure type \code{Term}.

It is not hard to show that if we had instead translated the rules $U_i$-\intro\ and $U_i$-\textsc{cumul} directly as
\[
\begin{prooftree}[center=false]
    \infer0[\code{U\_intro}]{\code{(U i)\Colon\,(U (S i))}}
\end{prooftree}
\qqquad
\begin{prooftree}[center=false]
    \hypo{\code{A\Colon\,(U i)}}
    \infer1[\code{U\_cumul}]{\code{A\Colon\,(U (S i))}}
\end{prooftree}
\]
then \code{U\_hierarchy} and \code{U\_cumulative} are derivable in $Pure + \ref{eq:Ord-rules} + \code{U\_intro} + \code{U\_cumul}$,
and \code{U\_intro}, \code{U\_cumul} are derivable in $Pure + \ref{eq:Ord-rules} + \code{U\_hierarchy} + \code{U\_cumulative}$.

\subsection{Type rules}

The basic type rules are implemented in their respective \code{.thy} files in \cite{isabelle-hott}.
\bigskip

\noindent \emph{Dependent product:}

{\small
\[
\begin{prooftree}
    \hypo{\code{A\Colon\,U i}}
    \hypo{\code{\All{}x. x\Colon\,A \Implies\ B x\Colon\,U i}}
    \infer2[\code{Prod\_form}]{\code{(\Prod{x\Colon\,A}{B x})\Colon\,U i}}
\end{prooftree}
\]

\[
\begin{prooftree}[center=false]
    \hypo{\code{A\Colon\,U i}}
    \hypo{\code{\All{}x. x\Colon\,A \Implies\ b x: B x}}
    \infer2[\code{Prod\_intro}]{\code{(\blm{}x. b x)\Colon\,\Prod{x\Colon\,A}{B x}}}
\end{prooftree}
\qqquad
\begin{prooftree}[center=false]
    \hypo{\code{f\Colon\,\Prod{x\Colon\,A}{B x}}}
    \hypo{\code{a\Colon\,A}}
    \infer2[\code{Prod\_elim}]{\code{f`a\Colon\,B a}}
\end{prooftree}
\]

\[
\begin{prooftree}[center=false]
    \hypo{\code{\All{}x. x\Colon\,A \Implies\ b x\Colon\,B x}}
    \hypo{\code{a\Colon\,A}}
    \infer2[\code{Prod\_appl}]{\code{(\blm{}x. b x)`a \Equiv\ b a}}
\end{prooftree}
\qqquad
\begin{prooftree}[center=false]
    \hypo{\code{f\Colon\,\Prod{x\Colon\,A}{B x}}}
    \infer1[\code{Prod\_uniq}]{\code{\blm{}x. (f`x) \Equiv\ f}}
\end{prooftree}
\]
}

Object-logic lambda terms are distinct from the lambda terms of Pure, and are written with boldface \blm\ instead of the regular $\lambda$.
We use the backtick \code{`} to denote object-level application of lambda terms.
\bigskip

\noindent \emph{Dependent sum:}

{\small
\[
\begin{prooftree}
    \hypo{\code{A\Colon\,U i}}
    \hypo{\code{\All{}x. x\Colon\,A \Implies\ B x\Colon\,U i}}
    \infer2[\code{Sum\_form}]{\code{(\Sum{x\Colon\,A}{B x})\Colon\,U i}}
\end{prooftree}
\]

\[
\begin{prooftree}
    \hypo{\code{\All{}x. x\Colon\,A \Implies\ B x\Colon\,U i}}
    \hypo{\code{a\Colon\,A}}
    \hypo{\code{b\Colon\,B a}}
    \infer3[\code{Sum\_intro}]{\code{<a,b>\Colon\,\Sum{x\Colon\,A}{B x}}}
\end{prooftree}
\]

\[
\begin{prooftree}
    \hypo{
    \begin{gathered} \textstyle
        \code{\All{}p. p\Colon\,\Sum{x\Colon\,A}{B x} \Implies\ C p\Colon\,U i} \\ \textstyle
        \code{\All{}x y. \llb{}x\Colon\,A; y\Colon\,B x\rrb\ \Implies\ f x y\Colon\,C <x,y>}
        \qquad
        \code{p\Colon\,\Sum{x\Colon\,A}{B x}}
    \end{gathered}
    }
    \infer1[\code{Sum\_elim}]{\code{ind\sub{\(\sum\)} f p\Colon\, C p}}
\end{prooftree}
\]

\[
\begin{prooftree}
    \hypo{
    \begin{gathered} \textstyle
        \code{\All{}p. p\Colon\,\Sum{x\Colon\,A}{B x} \Implies\ C p\Colon\,U i} \\ \textstyle
        \code{\All{}x y. \llb{}x\Colon\,A; y\Colon\,B x\rrb\ \Implies\ f x y\Colon\,C <x,y>}
        \qquad
        \code{\All{}x. x\Colon\,A \Implies\ B x\Colon\,U i} \\ \textstyle
        \code{a\Colon\,A} \qquad \code{b\Colon\,B a}
    \end{gathered}
    }
    \infer1[\code{Sum\_comp}]{\code{ind\sub{\(\sum\)} f p\Colon\, C p}}
\end{prooftree}
\]
}
\bigskip

\noindent \emph{Equality:}

{\small
\[
\begin{prooftree}[center=false]
    \hypo{\code{A\Colon\,U i}}
    \hypo{\code{a\Colon\,A}}
    \hypo{\code{b\Colon\,A}}
    \infer3[\code{Equal\_form}]{\code{a =\sub{A} b\Colon\,U i}}
\end{prooftree}
\qqquad
\begin{prooftree}[center=false]
    \hypo{\code{a\Colon\,A}}
    \infer1[\code{Equal\_intro}]{\code{(refl a)\Colon\,a =\sub{A} a}}
\end{prooftree}
\]

\[
\begin{prooftree}
    \hypo{
    \begin{gathered} \textstyle
        \code{\All{}x y p. \llb{}x\Colon\,A; y\Colon\,A; p\Colon\,x =\sub{A} y\rrb\ \Implies\ C x y p\Colon\,U i} \\ \textstyle
        \code{\All{}x. x\Colon\,A \Implies\ f x\Colon\,C x x (refl x)} \\ \textstyle
        \code{x\Colon\,A} \qquad \code{y\Colon\,A} \qquad \code{p\Colon\,x =\sub{A} y}
    \end{gathered}
    }
    \infer1[\code{Equal\_elim}]{\code{ind\sub{=} f p\Colon\, C x y p}}
\end{prooftree}
\]

\[
\begin{prooftree}
    \hypo{
    \begin{gathered} \textstyle
        \code{\All{}x y p. \llb{}x\Colon\,A; y\Colon\,A; p\Colon\,x =\sub{A} y\rrb\ \Implies\ C x y p\Colon\,U i} \\ \textstyle
        \code{\All{}x. x\Colon\,A \Implies\ f x\Colon\,C x x (refl x)} \qquad \code{a\Colon\,A}
    \end{gathered}
    }
    \infer1[\code{Equal\_comp}]{\code{ind\sub{=} f (refl a) \Equiv\ f a}}
\end{prooftree}
\]
}
\bigskip

\noindent \emph{Natural numbers:}

{\small
\[
\begin{prooftree}[center=false]
    \infer0[\code{Nat\_form}]{\code{\NN\Colon\,U O}}
\end{prooftree}
\qqquad
\begin{prooftree}[center=false]
    \infer0[\code{Nat\_intro\_0}]{\code{{\texttt 0}\Colon\,\NN}}
\end{prooftree}
\qqquad
\begin{prooftree}[center=false]
    \hypo{\code{n\Colon\,\NN}}
    \infer1[\code{Nat\_intro\_succ}]{\code{succ n\Colon\,\NN}}
\end{prooftree}
\]

\[
\begin{prooftree}
    \hypo{
    \begin{gathered} \textstyle
        \code{\All{}n. n\Colon\,\NN\ \Implies\ C n\Colon\,U i}
        \qquad
        \code{\All{}n c. \llb{}n\Colon\,\NN; c\Colon\,C n\rrb\ \Implies\ f n c\Colon\,C (succ n)} \\ \textstyle
        \code{a\Colon\,C 0} \qquad \code{n\Colon\,\NN}
    \end{gathered}
    }
    \infer1[\code{Nat\_elim}]{\code{ind\sub{\NN} f a n\Colon\,C n}}
\end{prooftree}
\]

\[
\begin{prooftree}
    \hypo{
    \begin{gathered} \textstyle
        \code{\All{}n. n\Colon\,\NN\ \Implies\ C n\Colon\,U i}
        \qquad
        \code{\All{}n c. \llb{}n\Colon\,\NN; c\Colon\,C n\rrb\ \Implies\ f n c\Colon\,C (succ n)}
        \qquad
        \code{a\Colon\,C 0}
    \end{gathered}
    }
    \infer1[\code{Nat\_comp\_0}]{\code{ind\sub{\NN} f a 0 \Equiv\ a}}
\end{prooftree}
\]

\[
\begin{prooftree}
    \hypo{
    \begin{gathered} \textstyle
        \code{\All{}n. n\Colon\,\NN\ \Implies\ C n\Colon\,U i}
        \qquad
        \code{\All{}n c. \llb{}n\Colon\,\NN; c\Colon\,C n\rrb\ \Implies\ f n c\Colon\,C (succ n)} \\ \textstyle
        \code{a\Colon\,C 0} \qquad \code{n\Colon\,\NN}
    \end{gathered}
    }
    \infer1[\code{Nat\_comp\_succ}]{\code{ind\sub{\NN} f a (succ n) \Equiv\ f n (ind\sub{\NN} f a n)}}
\end{prooftree}
\]
}
\bigskip

\noindent \emph{Coproduct:}

{\small
\[
\begin{prooftree}
    \hypo{\code{A\Colon\,U i}}
    \hypo{\code{B\Colon\,U i}}
    \infer2[\code{Coprod\_form}]{\code{A + B\Colon\,U i}}
\end{prooftree}
\]

\[
\begin{prooftree}[center=false]
    \hypo{\code{a\Colon\,A}}
    \hypo{\code{B\Colon\,U i}}
    \infer2[\code{Coprod\_intro\_inl}]{\code{inl a\Colon\,A + B}}
\end{prooftree}
\qqquad
\begin{prooftree}[center=false]
    \hypo{\code{A\Colon\,U i}}
    \hypo{\code{b\Colon\,B}}
    \infer2[\code{Coprod\_intro\_inr}]{\code{inr b\Colon\,A + B}}
\end{prooftree}
\]

\[
\begin{prooftree}
    \hypo{
    \begin{gathered} \textstyle
        \code{\All{}u. u\Colon\,A + B \Implies\ C u\Colon\,U i} \\ \textstyle
        \code{\All{}x. x\Colon\,A \Implies\ c x\Colon\,C (inl x)}
        \qquad
        \code{\All{}y. y\Colon\,B \Implies\ d y\Colon\,C (inr y)} \\ \textstyle
        \code{u\Colon\,A + B}
    \end{gathered}
    }
    \infer1[\code{Coprod\_elim}]{\code{ind\sub{+} c d u\Colon\,C u}}
\end{prooftree}
\]

\[
\begin{prooftree}
    \hypo{
    \begin{gathered} \textstyle
        \code{\All{}u. u\Colon\,A + B \Implies\ C u\Colon\,U i} \\ \textstyle
        \code{\All{}x. x\Colon\,A \Implies\ c x\Colon\,C (inl x)}
        \qquad
        \code{\All{}y. y\Colon\,B \Implies\ d y\Colon\,C (inr y)} \\ \textstyle
        \code{a\Colon\,A}
    \end{gathered}
    }
    \infer1[\code{Coprod\_comp\_inl}]{\code{ind\sub{+} c d (inl a) \Equiv\ c a}}
\end{prooftree}
\]

\[
\begin{prooftree}
    \hypo{
    \begin{gathered} \textstyle
        \code{\All{}u. u\Colon\,A + B \Implies\ C u\Colon\,U i} \\ \textstyle
        \code{\All{}x. x\Colon\,A \Implies\ c x\Colon\,C (inl x)}
        \qquad
        \code{\All{}y. y\Colon\,B \Implies\ d y\Colon\,C (inr y)} \\ \textstyle
        \code{b\Colon\,B}
    \end{gathered}
    }
    \infer1[\code{Coprod\_comp\_inl}]{\code{ind\sub{+} c d (inr b) \Equiv\ d b}}
\end{prooftree}
\]
}
\bigskip

\noindent \emph{Empty and unit types:}

{\small
\[
\begin{prooftree}[center=false]
    \infer0[\code{Empty\_form}]{\code{\(\bm{0}\)\Colon\,U O}}
\end{prooftree}
\qqquad
\begin{prooftree}[center=false]
    \hypo{\code{\All{}x. x\Colon\,\zero\ \Implies\ C x\Colon\,U i}}
    \hypo{\code{z\Colon\,\zero}}
    \infer2[\code{Empty\_elim}]{\code{ind\sub{\zero} z\Colon\,C z}}
\end{prooftree}
\]

\[
\begin{prooftree}[center=false]
    \infer0[\code{Unit\_form}]{\code{\one\Colon\,U O}}
\end{prooftree}
\qqquad
\begin{prooftree}[center=false]
    \infer0[\code{Unit\_intro}]{\code{\(\ast\)\Colon\,\one}}
\end{prooftree}
\]

\[
\begin{prooftree}
    \hypo{\code{\All{}x. x\Colon\,\one\ \Implies\ C x\Colon\,U i}}
    \hypo{\code{c\Colon\,C \(\ast\)}}
    \hypo{\code{a\Colon\,\one}}
    \infer3[\code{Unit\_elim}]{\code{ind\sub{\one} c a\Colon\,C a}}
\end{prooftree}
\]

\[
\begin{prooftree}
    \hypo{\code{\All{}x. x\Colon\,\one\ \Implies\ C x\Colon\,U i}}
    \hypo{\code{c\Colon\,C \(\ast\)}}
    \infer2[\code{Unit\_comp}]{\code{ind\sub{\one} c \(\ast\) \Equiv\ c}}
\end{prooftree}
\]
}

\subsection{Additional rules}

We axiomatize the additional congruence rule
\[
\begin{prooftree}
    \hypo{\code{A\Colon\,U i}}
    \hypo{\code{\All{}x. x\Colon\,A \Implies\ b x \Equiv\ c x}}
    \infer2[\code{Prod\_eq}]{\code{\blm{}x. b x \Equiv\ \blm{}x. c x}}
\end{prooftree}
\]
governing definitional equality of object lambda terms.

\subsection{Univalence}

The formulation of the inference rules of the previous sections is sufficient to encode the definitions and proofs needed to define univalence as presented in Section \ref{sec:HoTT}.
These are implemented in the theory file \code{Univalence.thy} in \cite{isabelle-hott}, and its dependencies \code{EqualProps.thy}, \code{ProdProps.thy}, and \code{Sum.thy}.
\bigskip

\noindent \emph{Univalence:}

\[
\begin{prooftree}
    \infer0[\code{UA}]{\code{univalence\Colon\,isequiv idtoeqv}}
\end{prooftree}
\]

\noindent where \code{isequiv} and \code{idtoeqv} are defined constants encoding $\isequiv$ and $\idtoeqv$ in Pure.

\section{Metatheory}

We discuss the metatheory of $HoTT_P$ as compared with that of $HoTT$.

As discussed earlier, formulating $HoTT_P$ as a polymorphic theory results in differences in the metatheory, in particular regard to normalization, and decidability of type checking and inference in the underlying implementation of Martin-L\"{o}f type theory
\[ ITT_P = HoTT_P - \code{UA}. \]
It is at the moment unclear to the author if $ITT_P$ enjoys the same properties as $ITT$, and what exactly is lost when removing annotations and introducing polymorphism.

There are, however, properties of $ITT$ that we have taken care to preserve.
For example, it is an important consequence of the normalization theorem in $ITT$ that if
\[ \turnstile a \colon A, \]
then
\[ \turnstile A \colon U_i \]
for some $i$ (Theorem 3.12, \cite{ml73}), i.e. well-typedness of terms implies well-formedness of their types.
The type rules of $HoTT_P$ have been deliberately formulated in order to maintain this property, which is the reason some of the rules in Section \ref{sec:HoTTP-rules} have additional premises not present in the corresponding $HoTT$ rules.
That is to say, by a structural induction on the depth of proof terms, we can show the following:

\begin{theorem}
If
\[ \Gamma \turnstile_{ITT_P} \code{a\Colon\,A} \]
then
\[ \Gamma \turnstile_{ITT_P} \code{A\Colon\,U i} \]
for some numeral \code{i}.
\end{theorem}

Using a similar technique, one should be able to show that if the judgment
\begin{pure}
\Prod{x\Colon\,A}{B x}\Colon\,U i
\end{pure}
is derivable in $ITT_P$, then so are the judgments
\begin{pure}
A\Colon\,U i
\end{pure}
and
\begin{pure}
\All{}x. x\Colon\,A \Implies B x\Colon\,U i\textrm{,}
\end{pure}
essentially because there is no other way the type \code{\Prod{x\Colon\,A}{B x}} can be formed.

Thus we conjecture

\begin{conjecture} \label{conj:wellformedness}
The inference rules

\[
\begin{prooftree}[center=false]
    \hypo{\code{\Prod{x\Colon\,A}{B x}\Colon\,U i}}
    \infer1[\code{Prod\_wellform1}]{\code{A\Colon\,U i}}
\end{prooftree}
\qqquad
\begin{prooftree}[center=false]
    \hypo{\code{\Prod{x\Colon\,A}{B x}\Colon\,U i}}
    \infer1[\code{Prod\_wellform2}]{\code{\All{}x. x\Colon\,A \Implies\ B x\Colon\,U i}}
\end{prooftree}
\]
are correct in $ITT_P$, as well as similar \emph{well-formedness rules} for the dependent sum, equality, and coproduct types.
\end{conjecture}

These type well-formedness rules are axiomatized in their respective theory files, and form a component of one of the automated proof methods implemented as part of $Isabelle/HoTT$.

\chapter{$Isabelle/HoTT$: an invitation} \label{ch:4}

We now come to the end of this thesis, which is really only just the beginning, as we now point the reader in the direction of the work that we have actually been doing all along---the implementation in code of the theory $HoTT$ as a working object logic in the Isabelle system.

There is much in the code that we have not been able to describe here: a collection of semi-automated proof methods that can be used to assist the derivation process, numerous examples illustrating its use to prove nontrivial statements and search for proof terms, formalizations of proofs from \cite{hottbook}, along with descriptive text explaining the details of the implementation in theory and practice.

As hinted at in the closing paragraphs of the previous chapter, there is still much work left to be done, both theoretically and practically.
It is however hoped that the work in this thesis will ultimately be of use in laying the foundations for bringing homotopy type theory to the Isabelle prover.
\bigskip

\noindent Code, issues, and documentation updated at \url{https://github.com/jaycech3n/Isabelle-HoTT}.

% \include{automated-reasoning}

% \include{some-notes}

% \include{further-work}

% \begin{appendices}

% \chapter{Theory file dependencies}

% \end{appendices}

\chapter*{Bibliography}
\addcontentsline{toc}{chapter}{Bibliography}

\printbibliography[heading=none]

\printindex

\end{document}